\documentclass[10pt]{article}

\usepackage{amsmath}
\usepackage{amsfonts}
\usepackage{amsthm}
\usepackage[english]{babel}
\usepackage{graphicx}
\usepackage[all]{xy}
\usepackage{amssymb}
\newtheorem{theorem}{Theorem}
\newtheorem{lemma}{Lemma}
\newtheorem{definition}{Definition}
\newtheorem{corollary}{Corollary}
\newtheorem{proposition}{Proposition}
\newtheorem{remark}{Remark}

\newcommand{\mR}{\mathbb{R}}
\newcommand{\mC}{\mathbb{C}}
\newcommand{\mN}{\mathbb{N}}
\newcommand{\mE}{\mathbb{E}}
\newcommand{\mZ}{\mathbb{Z}}
\newcommand{\mS}{\mathbb{S}}

\newcommand{\mK}{\mathbb{K}}

\newcommand{\mk}{\mathfrak{k}}

\newcommand{\mg}{\mathfrak{g}}

\newcommand{\cM}{\mathcal{M}}
\newcommand{\cH}{\mathcal{H}}

\newcommand{\cP}{\mathcal{P}}
\newcommand{\cC}{\mathcal{C}}

\newcommand{\cK}{\mathcal{K}}

\newcommand{\cO}{\mathcal{O}}

\newcommand{\cU}{\mathcal{U}}

\newcommand{\cJ}{\mathcal{J}}

\newcommand{\ux}{\underline{x}}

\newcommand{\uxb}{\underline{x} \grave{}}

\newcommand{\px}{\partial_{\bold{x}}}

\newcommand{\upx}{\partial_{\underline{x}}}

\newcommand{\osp}{\mathfrak{osp}(m|2n)}

\newcommand{\sC}{\cC l_{m|2n}}

\begin{document}
%%%%%%%%%%%%%%%%%%%%%%%%%%%%%%%%%%%%%%%%%%%%%%%%%%%%%%%%
\title{Conformal symmetries of the super Dirac operator}

\author{K.\ Coulembier\thanks{Ph.D. Fellow of the Research Foundation - Flanders (FWO), E-mail: {\tt Coulembier@cage.ugent.be}} $\quad$ H. De Bie\thanks{ E-mail: {\tt Hendrik.DeBie@ugent.be}}}

\date{\small{Clifford Research Group}\\
\small{Department of Mathematical Analysis\\
Faculty of Engineering and Architecture -- Ghent University\\ Krijgslaan 281, 9000 Gent,
Belgium}
}

\maketitle

\begin{abstract}
In this paper, the Dirac operator, acting on super functions with values in super spinor space, is defined along the lines of the construction of generalized Cauchy-Riemann operators by Stein and Weiss. The introduction of the superalgebra of symmetries $\mathfrak{osp}(m|2n)$ is a new and essential feature in this approach. This algebra of symmetries is extended to the algebra of conformal symmetries $\mathfrak{osp}(m+1,1|2n)$. The kernel of the Dirac operator is studied as a representation of both algebras. The construction also gives an explicit realization of the Howe dual pair $\mathfrak{osp}(1|2)\times\mathfrak{osp}(m|2n)\subset \mathfrak{osp}(m+4n|2m+2n)$. Finally, the super Dirac operator gives insight into the open problem of classifying invariant first order differential operators in super parabolic geometries.

\end{abstract}


%\tableofcontents

\section{Introduction}

This paper defines and studies the Dirac operator on flat superspace $\mR^{m|2n}$. In a more ad hoc approach in e.g. \cite{MR2374394} a related operator has been introduced of which the operator in this paper is an improvement. The definition of the operator in the current paper follows the ideas of Stein and Weiss in \cite{MR0223492} which allowed for a unified treatment of generalized Cauchy-Riemann operators on $\mR^m$. We also study the kernel of the super Dirac operator, its symmetries and the related Howe duality. In this introduction we give an overview of the motivations to study this particular operator.

The construction of the Dirac operator is part of a larger program to classify conformally invariant first order operators on superspace. The conformal Killing vector fields on $\mR^{m|2n}$ generate a Lie superalgebra isomorphic to the real form $\mathfrak{osp}(m+1,1|2n)$. This algebra appears as the algebra of conformal symmetries in super field theories in e.g. \cite{MR1626117}. We will prove that the super Dirac operator is a conformally invariant operator. The classification in superspace would be an extension of the classification of Fegan in \cite{MR0482879} of conformally invariant first order differential operators on $\mR^m$ or $\mS^m$. The result of Fegan has already been generalized to other parabolic geometries in \cite{MR1738448, Orsted, Soucek}.

The Dirac operator in this paper is an interesting example in this classification because it reveals two differences between the classical case and the case of supergeometry. First of all, the functions on which this operator acts take values in an infinite dimensional $\osp$-representation, whereas the classical classifications in \cite{MR0482879, MR1738448} only consider finite dimensional representations. As will become apparent in this paper, the natural extension of the classical Dirac operator leads to infinite dimensional representations. This is already the case for the (generalized) symplectic Dirac operators on $\mR^{2n}$ studied in \cite{MR2458281}. There, a class of infinite dimensional representations was included in the classification of invariant operators on metaplectic contact projective geometries. The second reason why this Dirac operator is an interesting starting point in the classification is that in some cases the construction encounters the problem of 
 tensor products which are not completely reducible. An important step in the work of Fegan is the decomposition of the tensor product $\mC^m\otimes V$, with $V$ an irreducible $\mathfrak{so}(m)$-representation, into irreducible representations. The corresponding tensor product for $\osp$ is not necessarily completely reducible. Therefore, the Dirac operator in the current paper shows how this additional difficulty can be approached and how the classification can be expected to differ from the classical case. In Section \ref{Fegan} the conclusions which are made throughout the paper towards a Fegan classification will be summarized.

Another motivation comes from the theory of Lie superalgebra representations, where the question of irreducibility of indecomposable highest weight representations plays a central role. In \cite{OSpHarm} the kernel of the super Laplace operator was studied as an $\osp$-representation. This gave new important information on a certain class of finite dimensional irreducible highest weight representations. Similarly, in the current paper, the kernel of the Dirac operator leads to an explicit realization of some infinite dimensional highest weight representations of $\osp$. This direct approach leads to a better understanding of the concept of indecomposable but irreducible highest weight modules, see \cite{MR051963}. In particular we will also obtain the full decomposition series of a certain tensor product. This was already studied in \cite{Tensor}, but the completeness of the decomposition could not be proved there.

The Dirac operator on $\mR^{m|2n}$ generates the Lie superalgebra $\mathfrak{osp}(1|2)$ together with the corresponding vector variable. These operators commute with the action of the Lie superalgebra $\osp$. We prove that this is an explicit construction of the spinor representation of $\mathfrak{osp}(m+4n|2m+2n)$, which according to \cite{MR1893457} is a Lie superalgebra in which $\mathfrak{osp}(1|2)$ and $\mathfrak{osp}(m|2n)$ are each other's centralizers. This leads to a realization of the Howe duality $(\mathfrak{osp}(m|2n),\mathfrak{osp}(1|2))$ if $m-2n\not\in-2\mN$, which has not been studied before. The explicit Howe duality is summarized in the subsequent equations \eqref{sumHowe1}, \eqref{sumHowe2} and \eqref{sumHowe3}. The Howe dualities $\mathfrak{osp}(1|2)\times \mathfrak{so}(m)\subset\mathfrak{osp}(m|2m)$ in \cite{Kyo} and $\mathfrak{osp}(1|2)\times \mathfrak{sp}(2n)\subset\mathfrak{osp}(4n|2n)$ in \cite{Krysl} are limit cases of the Howe duality in the current
  paper.
  If $m-2n\in-2\mN$ the realization of the Howe duality breaks down and this different behavior is studied in Section \ref{Mkasreps}. The Howe duality corresponding to the super Laplace operator, which is the square of the super Dirac operator, is $\mathfrak{sl}_2\times\mathfrak{osp}(m|2n)\subset\mathfrak{osp}(4n|2m)$ and was studied in \cite{OSpHarm}.

Finally, the super Dirac operator in the current paper also unifies two classical operators. These are the Dirac operator on $\mR^m$, see e.g. \cite{MR1169463, Kyo} and an operator constructed in \cite{Krysl} used to study differential forms on $\mR^{2n}$ with values in the Kostant symplectic spinors $\mS_{0|2n}$, as is presented in Figure \ref{tabelhowe}. In case we only consider polynomials, the classical Dirac operator acts on the symmetric tensor powers of the fundamental $\mathfrak{so}(m)$-representation $\mC^{m}$ with values in the orthogonal spinors $\mS_{m}$. The Dirac operator generates the Lie superalgebra $\mathfrak{osp}(1|2)$ together with the vector variable and this algebra commutes with the action of $\mathfrak{so}(m)$. Since the decomposition into irreducible representations of $S(\mC^m)\otimes \mS_{m}$ under the joint action of $\mathfrak{osp}(1|2)\times \mathfrak{so}(m)$ is multiplicity-free, we obtain a realization of the Howe dual pair 
 $\left( \mathfrak{so}(m),\mathfrak{osp}(1|2)\right)$. In \cite{Krysl},  a similar construction was made for the action of $\mathfrak{sp}(2n)$ on differential forms on $\mR^{2n}$ (the outer power $\Lambda(\mC^{2n})$) with values in the symplectic spinors $\mS_{0|2n}$. The commutant of $\mathfrak{sp}(2n)$ was given by $\mathfrak{osp}(1|2)$ and one of the generators of this Lie superalgebra can be seen as an analogue of the Dirac operator. This construction showed that $\left( \mathfrak{sp}(2n),\mathfrak{osp}(1|2)\right)$ is a Howe dual pair for this realization. The super Dirac operator we will construct on $\mR^{m|2n}$, reduces in the two limiting cases $m=0$ and $n=0$ to one of these situations. The polynomials on $\mR^{m|2n}$ are given by the supersymmetric tensor power $S(\mC^{m|2n})$, which corresponds to $S(\mC^m)\otimes \Lambda(\mC^{2n})$, and the super spinor space $\mS_{m|2n}$ generalizes and contains both the orthogonal and symplectic spinors, see \cite{Tensor}. 

The super Dirac operator also fits into a bigger picture with the symplectic Dirac operator on $\mR^{2n}$ of \cite{DBSS}. Figure \ref{tabelhowe} also contains this symplectic Dirac operator, acting on functions on $\mR^{2n}$ with values in the symplectic spinors $\mS_{0|2n}$. The symplectic Dirac operator and the corresponding vector variable generate the Lie algebra $\mathfrak{sl}_2$, which leads to the Howe dual pair $(\mathfrak{sp}(2n),\mathfrak{sl}_2)$. It seems plausible that it is possible to generalize this operator to superspace as well. This should also lead to a generalization of the operator appearing in the study of differential forms on $\mR^m$ with values in the spinor space $\mS_m$, see \cite{MR1169463, MR1368704}, since the algebra of super differential forms on $\mR^{m|2n}$ contains a commuting subalgebra isomorphic to the polynomials on $\mR^{2n}$ as well as the differential forms on $\mR^m$.

\begin{figure}
\caption{Howe dualities for the super Dirac operator and limiting cases}
\label{tabelhowe}
\[
\xymatrix{
& *+[F]\txt{$S(\mC^{m|2n})\otimes \mS_{m|2n}$\\ $\mathfrak{osp}(1|2)\times \mathfrak{osp}(m|2n)$}  \ar[ddl]_{n \rightarrow 0}\ar[ddr]^{m \rightarrow 0}&\\
\\
*+[F]\txt{$S(\mC^m)\otimes \mS_{m}$\\ $\mathfrak{osp}(1|2)\times \mathfrak{so}(m)$ \\ \cite{MR1169463},  \cite{Kyo}}\ar@{<~>}[d]&&*+[F]\txt{$\Lambda(\mC^{2n})\otimes \mS_{0|2n}$ \\  $\mathfrak{osp}(1|2)\times \mathfrak{sp}(2n)$ \\ \cite{Krysl}}\ar@{<~>}[d]\\
*+[F]\txt{$\Lambda(\mC^m)\otimes \mS_{m}$ \\ $\mathfrak{sl}(2)\times \mathfrak{so}(m)$ \\ \cite{MR1368704}}&&*+[F]\txt{$S(\mC^{2n})\otimes \mS_{0|2n}$ \\ $\mathfrak{sl}(2)\times \mathfrak{sp}(2n)$ \\ \cite{DBSS}}\\
\\
& *+[F]\txt{$\Lambda(\mC^{m|2n})\otimes \mS_{m|2n}$ \\$\mathfrak{sl}(2)\times \mathfrak{osp}(m|2n)$}  \ar@{-->}[uur]_{m \rightarrow 0}\ar@{-->}[uul]^{n \rightarrow 0}&\\
}
\]
\end{figure}
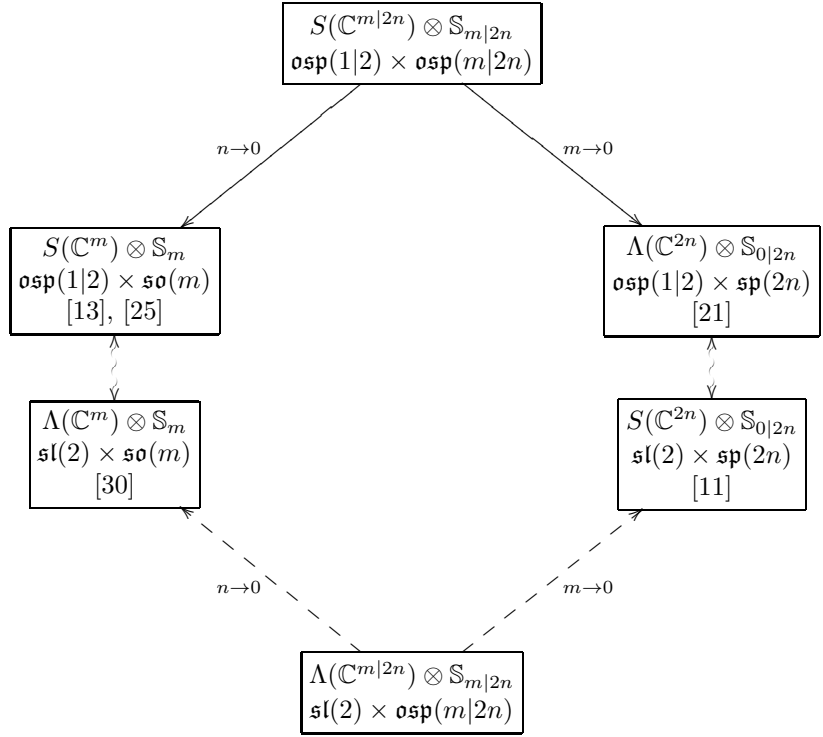

The paper is organized as follows. In the preliminaries we recall the basic notions of the classical Dirac operator as a Stein-Weiss type operator and the necessary results on harmonic analysis on $\mR^{m|2n}$ and super spinor space. Then we introduce the super Clifford algebra and relate it to the super spinor space and the Lie superalgebra $\mathfrak{osp}(m|2n)$. This gives the necessary tools to define and elegantly describe the super Dirac operator. At the end of Section \ref{secdefDirac} we show a similarity between the super Dirac operator and a realization of $\mathfrak{osp}(1|2)$ due to Bernstein. In Section \ref{secHowe} the Howe dual pair $(\osp,\mathfrak{osp}(1|2))$ is studied, where the Lie superaglebra $\mathfrak{osp}(1|2)$ is generated by the super Dirac operator and the vector variable. In particular the Fischer decomposition of $S(\mC^{m|2n})\otimes\mS_{m|2n}$ is obtained if $m-2n\not\in-2\mN$. In Section \ref{Mkasreps} the properties of the kernel of the Dira
 c operat
 or as an $\mathfrak{osp}(m|2n)$-representation are studied. In Section \ref{secCon} we construct all first order generalized symmetries of the super Dirac operator which have a scalar leading term. They generate the Lie superalgebra $\mathfrak{osp}(m+1,1|2n)$, which is isomorphic to the Lie superalgebra of conformal Killing vector fields on $\mR^{m|2n}$. Moreover, it is shown that the leading terms are exactly the conformal Killing vector fields. Then we study the kernel of the Dirac operator as an $\mathfrak{osp}(m+1,1|2n)$-representation. Finally, in Section \ref{Fegan}, we review the obtained insights towards the classification of conformally invariant first order differential operators.

\section{Preliminaries}
\label{preliminaries}
In this section we will recall some known facts about the classical Dirac operator, harmonic analysis in superspace and the super spinor space. First we mention some conventions and notations. Unless it is explicitly mentioned otherwise we will assume that $m>2$ holds.

The real orthogonal algebra will be denoted by $\mathfrak{so}(m)=\mathfrak{so}(m;\mR)$ and the irreducible $\mathfrak{so}(m)$-representa-tion with highest weight $\mu$ by $L^m_{\mu}$. The real orthosymplectic Lie superalgebra is denoted by $\mathfrak{osp}(m|2n)$. Its irreducible highest weight representations will be denoted by $K_\lambda^{m|2n}$, where $\lambda$ is the highest weight corresponding to the simple root system used in \cite{OSpHarm} and \cite{Tensor}. This root system differs from the distinguished one in \cite{MR051963} but is much more convenient to describe the type of representations we study. At the end of Section \ref{secHowe} the most relevant representations will also be expressed in terms of the distinguished root system. For general representations the procedure to calculate the highest weight is described in Section 4 in \cite{Tensor}. The roots of $\mathfrak{so}(m)$ are expressed in terms of $\epsilon_j$, $j=1,\cdots,\lfloor m/2\rfloor$ and those of 
 the symp
 lectic algebra $\mathfrak{sp}(2n)$ in terms of $\delta_i$, $i=1,\cdots,n$. Some important fundamental weights are given by $\omega_d=\frac{1}{2}(\epsilon_1+\epsilon_2+\cdots+\epsilon_d)$ if $m=2d$ or $m=2d+1$, $\omega_{d-1}=\frac{1}{2}(\epsilon_1+\cdots+\epsilon_{d-1}-\epsilon_d)$ if $m=2d$ and $\nu_j=\delta_1+\cdots+\delta_j$ for $1\le j\le n$.

\subsection{Dirac operator on $\mR^m$}
\label{secclassDi}

In this subsection we recall the basic notions concerning Clifford analysis on $\mR^m$, see \cite{MR1169463, Kyo, MR0223492}. No proofs will be given because they correspond to the limit case $n\to0$ of the Dirac operator on $\mR^{m|2n}$ studied in the current paper.

The complex Clifford algebra $\cC l_m$ corresponding to the vector space $\mC^m$ is generated by the standard basis vectors $e_j,$ $j=1,\cdots, m$ with commutation relations
\begin{equation}
\label{classClifford}
\quad e_je_k+e_ke_j=-2\delta_{jk}.
\end{equation}

\vspace{1mm}

The embedding $\mR^m\subset \cC l_m$ is given by identifying the vector $(x_{1}, \ldots, x_{m})$ with the vector variable $\ux=\sum_{j=1}^mx_je_j$. The variables $x_j$ are assumed to commute with the vectors $e_j$. The Dirac operator is given by
\begin{eqnarray}
\label{ClassDirac}
\upx=\sum_{j=1}^me_j\partial_{x_j}\quad
\end{eqnarray}
and acts on smooth functions with values in the Clifford algebra or a minimal left ideal of the Clifford algebra, the spinor space $\mS_m$:
\[\upx:\cC^\infty(\mR^m)\otimes \mS_m\to\cC^\infty(\mR^m)\otimes \mS_m.\]
All minimal left ideals in $\cC l_m$ are isomorphic to the spinor space. If $m=2d+1$, the spinor space is $\mS_m\cong L^m_{\omega_d}$ as an $\mathfrak{so}(m)$-representation. If $m=2d$ is even, the spinor space decomposes into two irreducible $\mathfrak{so}(m)$-representations, $\mS_m=\mS_m^+\oplus\mS_m^-\cong L^m_{\omega_d}\oplus L^m_{\omega_{d-1}}$. The square of the Dirac operator is given by the scalar Laplace operator
\begin{eqnarray*}
\upx^2=-\Delta_b=-\sum_{j=1}^m\partial_{x_j}^2.
\end{eqnarray*}

The Dirac operator can equally be constructed as a Stein-Weiss type operator, see \cite{MR0223492}. Consider the space $\cC^\infty(\mR^m)\otimes V$ of functions on $\mR^m$ with values in a simple $\mathfrak{so}(m)$-module $V$. The gradient can naturally be seen as an operator
\[\nabla: \cC^\infty(\mR^m)\otimes V\to\cC^\infty(\mR^m)\otimes \mC^m\otimes V,\]
which in coordinates is given by $\nabla f=\sum_{i=1}^m \partial_{x_j}f\otimes e_j$ in case $V$ is the trivial representation.

The decomposition into irreducible $\mathfrak{so}(m)$-representations of the tensor product of two irreducible representations $V$ and $W$ is given by $V\otimes W=V\boxtimes W\oplus\left(\oplus_{i}U_i\right)$ for some irreducible highest weight representations $U_i$ with highest weight lower than the sum of the highest weights of $V$ and $W$, and $V\boxtimes W$ the irreducible representation with highest weight equal to this sum. This last representation is called the Cartan product of the representations $V$ and $W$. When one of the two representations is the fundamental representation, the $\mathfrak{so}(m)$-invariant projection onto the Cartan product is denoted by
\[E:\mC^m\otimes V \to \mC^m\boxtimes V.\]
Stein and Weiss showed that many generalized Cauchy-Riemann systems with interesting properties correspond to operators of the form $\partial=E^\perp\circ \nabla$ with $E^\perp=1-E$, the invariant projection onto everything except the Cartan product, for some module $V$. The generalized Cauchy-Riemann system is then given by
\[\partial f=0\qquad\mbox{for}\quad f\in\cC^\infty(\mR^m)\otimes V\]
and functions satisfying this are called hyperholomorphic functions.

If the irreducible module $V$ corresponds to spinor spaces, the tensor product decomposition is given by
\begin{eqnarray*}
\mC^m\otimes\mS_m&\cong & \mC^m\boxtimes\mS_m\,\oplus\, \mS_m\qquad\mbox{for $m$ odd}\\
\mC^m\otimes\mS_m^+&\cong & \mC^m\boxtimes\mS^+_m\,\oplus\, \mS_m^-\qquad\mbox{for $m$ even}\\
\mC^m\otimes\mS_m^-&\cong & \mC^m\boxtimes\mS^-_m\,\oplus\, \mS_m^+\qquad\mbox{for $m$ even},
\end{eqnarray*}
see \cite{MR0482879, MR0223492}. In case $m$ is even we thus obtain two operators, acting between the spaces $\cC^\infty(\mR^m)\otimes\mS^+_m$ and $\cC^\infty(\mR^m)\otimes\mS^-_m$, which sum up to a single differential operator acting inside $\cC^\infty(\mR^m)\otimes\mS_m$.

%To obtain the most elegant expression, the definition of the generalized Cauchy-Riemann operator has to be subtly adjusted to $\partial =\phi^{-1}\circ E^\perp\circ\nabla$ with $\phi: \mS_m\to \mS_m\otimes\mC^m$ the injective homomorphism. The case $m$ even is similar, but there is a Dirac operator 

In the Stein-Weiss definition of the Dirac operator the Clifford algebra does not appear. However, the Dirac operator is most elegantly described by introducing the endomorphism algebra of the spinor space, which is the corresponding Clifford algebra, $\cC l_m\cong End(\mS_m)$. By using this Clifford algebra we can identify $\partial$ with the Dirac operator $\upx$ in equation \eqref{ClassDirac}. Note that the fact that $\mS_m$ can be realized as a left ideal in $\cC l_m$ and the fact that the action of $\cC l_m$ on $\mS_m$ is given by left-multiplication is not important in the Stein-Weiss construction. It will also no longer hold for the super Dirac operator, because the symplectic spinors are no left ideal inside the Weyl algebra. This phenomenon also appears in the case of the symplectic Dirac operator in \cite{DBSS} and in the study of the Howe dual pair on differential forms with values in the symplectic spinors in \cite{Krysl}.

On the sphere $\mS^m$, the conformal compactification of $\mR^m$, the Dirac operator is conformally invariant, see \cite{MR0482879}. This means that the Dirac operator acting on the space of functions $\Gamma(\mS^m,SO(m+1,1)\times_P \mS_m)$ is $SO(m+1,1)$-invariant. Here $P$ is the Poincar\'e group, which contains $SO(m)$, the translations which act trivially on $\mS_m$ and the rescaling which acts on $\mS^m$ through the conformal weight. The sphere satisfies $\mS^m\cong SO(m+1,1)/P$. The infinitesimal action of the morphisms on $\mS^m$ which preserve the class of metrics are the conformal Killing vector fields. They constitute the Lie algebra $\mathfrak{so}(m+1,1)$. These do not integrate to global diffeomorphisms when only the flat space $\mR^m\subset \mS^m$ is considered. The conformal invariance of the Dirac operator on $\mR^m$ is therefore expressed by considering the conformal Killing vector fields. This means that there is a set of first order differential operators $D
 $, which
  constitute the Lie algebra $\mathfrak{
 so}(m+1,1)$, for which a second differential operator $\delta$ exists such that
\[ \upx D= \delta \upx\]
holds. Such differential operators $D$ are called (generalized) symmetries. 

The Lie algebra $\mathfrak{so}(m)$ acts on spinor space through its realization as bivectors in the Clifford algebra $\cC l_m$. The standard generators which are realized on functions as $L_{ij}=x_i\partial_{x_j}-x_j\partial_{x_i}$  are given by $-\frac{1}{2}e_ie_j$ for $1\le i<j\le m$. The $\mathfrak{so}(m)$-action on $\cC^\infty(\mR^m)\otimes \cC l_m$ or $\cC^\infty(\mR^m)\otimes \mS_m$ is therefore given by the differential operators
\begin{eqnarray*}
x_i\partial_{x_j}-x_j\partial_{x_i}-\frac{1}{2}e_ie_j.
\end{eqnarray*}

The Dirac operator $\upx$ and the vector variable $\ux$ generate the Lie superalgebra $\mathfrak{osp}(1|2)$ and commute with the realization of $\mathfrak{so}(m)$ given above, see e.g. \cite{MR2374394, Kyo}. This is a consequence of straightforward commutation relations such as $\upx \ux+\ux\upx=-2\mE-m$, where $\mE=\sum_{j=1}^mx_j\partial_{x_j}$ is the Euler operator.

The solutions of the Dirac equation $\upx f=0$ are called monogenic functions. In particular, the space of spherical monogenics of degree $k$ is given by
\[
\cM^b_k=\{p\in \mC[x_1,\cdots,x_m]\otimes \mS_{m}|\,\mE p=k p\mbox{ and }\upx p=0\}.
\]

Arbitrary polynomials can be decomposed into a sum of products of the powers of the vector variable with spherical monogenics. This is the subject of the following decomposition, called monogenic Fischer decomposition by analogy with the Fischer decomposition of scalar polynomials based on harmonic functions.
\begin{theorem}
\label{bosmon}
The monogenic Fischer decomposition on $\mR^m$ is given by
\[
\mC[x_1,\cdots,x_m]\otimes \mS_{m}=\bigoplus_{j=0}^\infty\bigoplus_{k=0}^\infty \ux^j \cM^b_k.
\]
\end{theorem}
Let us discuss the consequences of this theorem in some more detail. First, we look at the case where $m=2d+1$ is odd. Each space $\cM^b_k$ is an irreducible representation for $\mathfrak{so}(m)$ with highest weight $k\epsilon_1+\omega_d$. Therefore each representation $L^{m}_{k\epsilon_1+\omega_d}$ appears an infinite amount of times in the decomposition. The corresponding isotypical component $\bigoplus_{j=0}^\infty \ux^j \cM^b_k$ corresponds to an irreducible $\mathfrak{osp}(1|2)$-module with weight vectors $\ux^j \cM^b_k$ and lowest weight $k+m/2$. Theorem \ref{bosmon} thus implies that under the joint action of $\mathfrak{osp}(1|2)\times\mathfrak{so}(m)$, the space $\mC[x_1,\cdots,x_m]\otimes\mS_{m}$ has a multiplicity-free irreducible direct sum decomposition. Additionally, each $\mathfrak{osp}(1|2)$-representation is paired up with only one $\mathfrak{so}(m)$-representation which also appears only once. This implies that $(\mathfrak{so}(m),\mathfrak{osp}(1|2))$ is a Howe dual pair (see \cite{MR0986027}) for the action on $\mC[x_1,\cdots,x_m]\otimes\mS_{m}$. At the beginning of Section \ref{secHowe} an overview is given about how the Howe duality of the Dirac operator, super Laplace operator and super Dirac operator extend the Howe duality of the Laplace operator included in \cite{MR0986027}.

The case $m=2d$ is slightly more complicated. By introducing the polynomials of even and odd degree with notation $\mC[x_1,\cdots,x_m]^\pm$, the decomposition in Theorem \ref{bosmon} can be refined to
\begin{eqnarray*}
\left(\mC[x_1,\cdots,x_m]^+\otimes \mS_{m}^+\right)\oplus\left(\mC[x_1,\cdots,x_m]^-\otimes \mS_{m}^-\right)&=&\bigoplus_{j=0}^\infty\bigoplus_{k=0}^\infty \ux^j \cM^b_k{}^{(-1)^{k}}\\
\left(\mC[x_1,\cdots,x_m]^-\otimes \mS_{m}^+\right)\oplus\left(\mC[x_1,\cdots,x_m]^+\otimes \mS_{m}^-\right)&=&\bigoplus_{j=0}^\infty\bigoplus_{k=0}^\infty \ux^j \cM^b_k{}^{(-1)^{k+1}},
\end{eqnarray*}
with 
\begin{eqnarray*}
\cM^b_k{}^{\pm}=\{p\in \mC[x_1,\cdots,x_m]\otimes \mS_{m}^\pm|\,\mE p=k p\mbox{ and }\upx p=0\}
\end{eqnarray*}
irreducible $\mathfrak{so}(m)$-representations with highest weight respectively given by $k\epsilon_1+\omega_d$ and $k\epsilon_1+\omega_{d-1}$. Now these two decompositions correspond to multiplicity-free irreducible direct sum decompositions, with one-to-one pairing, under the joint action of $\mathfrak{osp}(1|2)\times\mathfrak{so}(m)$.

\begin{remark}
{\rm Clearly, there is also a Howe duality corresponding to the Laplace operator. This is the well-known representation of $\mathfrak{sl}_2\times\mathfrak{so}(m)$ on $\mC[x_1,\cdots,x_m]$, see \cite{MR0986027}. This Howe duality has been generalized to $\mR^{m|2n}$ in \cite{OSpHarm} and the corresponding Fischer decomposition will be recalled in the subsequent Lemma \ref{superFischerLemma}.

Moreover, other Dirac-type operators and their Howe duals can be found in \cite{MR2677004, H12}.}
\end{remark}

\subsection{Super vector spaces and $\osp$}
\label{prelosp}

The standard basis of the graded vector space $V=\mK^{m|2n}$ (with $\mK$ a field which in this paper will always be $\mR$ of $\mC$) consists of the vectors $E_j$ for $1\le j\le m+2n$, where $E_j=(0,\cdots,0,1,0,\cdots,0)$ with $1$ at the $j$-th position. The elements $E_j$ with $1\le j\le m$ span $V_{\overline{0}}$, and $E_j$ with $m<j\le m+2n$ span $V_{\overline{1}}$. As a vector space $\mK^{m|2n}$ is clearly isomorphic to $\mK^{m+2n}$.

For any $\mZ_2$-graded vector space $V=V_{\overline{0}}\oplus V_{\overline{1}}$, a vector $u$ belonging to $V_{\overline{0}}\cup V_{\overline{1}}$ is called homogeneous, and in this case we define $|u|=\alpha$ for $u\in V_\alpha$ where $\alpha\in\mZ_2=\mZ/(2\mZ)$. We also introduce a function
\begin{eqnarray}
\label{gradmap}
[\cdot]:\{1,2,\cdots,m+2n\}\to\mZ_2, \quad
\text{$[j]=\overline{0}$ if $j\le m$ and $[j]=\overline{1}$ otherwise}.
\end{eqnarray}
Then $|E_j|=[j]$ for all $j$ for the super vector space $V=\mK^{m|2n}$.

The space of endomorphisms on $\mR^{m|2n}$ is denoted by End$(\mR^{m|2n})$ when seen as an associative algebra or by $\mathfrak{gl}(m|2n;\mR)$ when seen as a Lie superalgebra. As a vector space, End$(\mR^{m|2n})$ is isomorphic to End$(\mR^{m+2n})$. The grading on End$(\mR^{m|2n})$ is inherited naturally from the grading on $\mR^{m|2n}$. The super Lie bracket on $\mathfrak{gl}(m|2n;\mR)$ is given by $[A,B]=A\circ B-(-1)^{|A||B|}B\circ A$. We will always use this notation $[\cdot,\cdot]$, also in case $A$ and $B$ are odd and the super commutator equals the anti-commutator $\{\cdot,\cdot\}$. The bracket $[\cdot,\cdot]$ is super antisymmetric and satisfies a super Jacobi identity. A super vector space $V$ with such a super Lie bracket is called a Lie superalgebra if the bracket is grade-preserving, $[V_i,V_j]\subset V_{i+j}$. 

In this paper, the orthosymplectic metric $g\in\mR^{(m+2n)\times (m+2n)}$ is given by
\begin{eqnarray}
\label{gmetric}
g=\left( \begin{array}{cc} I_m&0\\ \vspace{-3.5mm} \\0&J_{2n}
\end{array}
 \right)&\mbox{with}&J_{2n}=\left( \begin{array}{cc} 0&I_{n}\\  \vspace{-3.5mm} \\-I_n&0
\end{array}
 \right).
 \end{eqnarray}
%This metric satisfies $\sum_{j=1}^{m+2n}g_{ij}g_{kj}=\delta_{ik}$ and $\sum_{j=1}^{m+2n}g_{ij}g_{jk}=(-1)^{[i]}\delta_{ik}$.
 
 The Lie superalgebra $\mathfrak{osp}(m|2n)$ can be defined as the subsuperalgebra of $\mathfrak{gl}(m|2n;\mR)$ that preserves this metric. Considering the applications in the current paper it is more natural to introduce $\osp$ through the standard generators, which constitute a subset of $\mathfrak{gl}(m|2n)$. The defining representation of $\mathfrak{osp}(m|2n)$ on $\mR^{m|2n}$ is given by
\begin{equation}
\label{nataction}
\cK_{ij}E_k=g_{kj}E_i-(-1)^{[i][j]}g_{ki}E_j.
\end{equation}
The operators $\cK_{ij}$ generate $\osp$, and satisfy the following super commutator relations:
\begin{eqnarray}
\nonumber
[\cK_{ij},\cK_{kl}]&=&\cK_{ij}\cK_{kl}-(-1)^{([i]+[j])([k]+[l])}\cK_{kl}\cK_{ij}\\
\label{commL}
&=&g_{kj}\cK_{il}+(-1)^{[i]([j]+[k])}g_{li}\cK_{jk}-(-1)^{[k][l]}g_{lj}\cK_{ik}-(-1)^{[i][j]}g_{ki}\cK_{jl}.
\end{eqnarray}
We will always assume real Lie superalgebras acting on complex spaces from now one. The defining representation satisfies $\mC^{m|2n}\cong K_{\epsilon_1}^{m|2n}$.

The tensor product $V\otimes W$ of two $\mathfrak{osp}(m|2n)$-representations $V$ and $W$, is again a representation with action defined by
\[
X\cdot (v\otimes w)=(X\cdot v)\otimes w\,+\, (-1)^{|X||v|}v\otimes (X\cdot w),
\]
for $X\in\mathfrak{osp}(m|2n)$, $v\in V$ both homogeneous and $w\in W$. The supersymmetric tensor product $V\odot V$ is the span in $V\otimes V$ of the elements $u\otimes v+(-1)^{|u||v|}v\otimes u$ for $u,v\in V$ homogeneous. This is a subrepresentation of $V\otimes V$.

\subsection{Harmonic analysis on $\mR^{m|2n}$}

In this subsection we recall some results on the study of the Laplace operator on superspace, see \cite{OSpHarm, DBE1}.

Superspaces are spaces where one considers not only commuting (bosonic) but also anti-commuting (fermionic) co-ordinates. The $2n$ anti-commuting variables ${x\grave{}}_i$ generate the complex Grassmann algebra $\Lambda_{2n}$. We consider a space with $m$ bosonic variables. The supervector $\bold{x}$ is defined to be
\[
\bold{x}=(X_1,\cdots,X_{m+2n})=(\ux,\uxb)%=(x_1,\cdots,x_m,{x\grave{}}_1,\cdots,{x\grave{}}_{2n}).
\]
The first $m$ variables are commuting and the last $2n$ anti-commuting. The commutation relations are then summarized in 
\[X_iX_j=(-1)^{[i][j]}X_jX_i \mbox{ for }i,j=1,\cdots,m+2n.\] 
The algebra generated by the variables $X_j$ is denoted by $\cP$ and is isomorphic to the supersymmetric tensor power of $\mC^{m|2n}$. The flat supermanifold, corresponding with these variables, is denoted by $\mR^{m|2n}$. The full algebra of functions on this supermanifold is $\cO(\mR^{m|2n})=\cC^\infty(\mR^m)\otimes \Lambda_{2n}$ which contains $\cP$ as a subalgebra. The partial derivatives are defined by the relation \[\partial_{X_j}X_k=\delta_{jk}+(-1)^{[j][k]}X_k\partial_{X_j}.\]

Using the orthosymplectic metric $g$ we can define the super Laplace operator and norm squared on $\mR^{m|2n}$, along with the Euler operator:
 \begin{equation}
 \label{defsuperLap}
 \Delta=\sum_{j,k=1}^{m+2n}\partial_{X_j}g_{jk}\partial_{X_k}, \qquad  R^2=\sum_{j,k=1}^{m+2n}X_jg_{jk}X_k,\qquad \mE=\sum_{j=1}^{m+2n}X_j\partial_{X_j}.
 \end{equation}
As in the classical case $\Delta$, $R^2$ and $\mE+\frac{m-2n}{2}$ generate the Lie algebra $\mathfrak{sl}_2$, see \cite{DBE1, MR0695958}. In these formulas, the parameter $M=m-2n$ replaces the classical dimension $m$. It turns out that $M$ plays an important role and will often characterize properties independently of the exact super dimension $m|2n$.

%As a result of the appearance of the parameter $M=m-2n$ where classically the dimension $m$ appears, this parameter plays an important role and will often characterize properties independently of the exact super dimension $m|2n$.
 
The null-solutions of the super Laplace operator are called (super) harmonic functions.
\begin{definition}
{\rm The space of} spherical harmonics {\rm of homogeneous degree $k$ is given by $\cH_k=\ker\Delta\cap\cP_k$, with $\cP_k$ the polynomials of degree $k$, i.e. those satisfying $\mE P=kP$.}
\end{definition}

We can use the metric to raise indices as $\partial_{X^j}=\sum_k g_{kj}\partial_{X_k}$. These partial derivatives satisfy $\partial_{X^j}R^2=2X_j$.

Since $\cP\cong \oplus_{k=0}^\infty S_k\left(\mC^{m|2n}\right)=S\left(\mC^{m|2n}\right)$, the $\osp$-action on $\cP$ is given by
\begin{eqnarray}
\label{Lij}
\pi_{\cO}:\cK_{ij}&\rightarrow& L_{ij}=X_i\partial_{X^j}-(-1)^{[i][j]}X_j\partial_{X^i}.
\end{eqnarray}
This action clearly extends from $\cP$ to the full algebra $\cO(\mR^{m|2n})$. The Laplace operator and norm squared commute with these differential operators. The actions of $\mathfrak{sl}_2$ and $\mathfrak{osp}(m|2n)$ on $\mR^{m|2n}$ therefore commute with each other. In \cite{OSpHarm} it was proved that this pair $(\mathfrak{osp}(m|2n),\mathfrak{sl}_2)$ constitutes a Howe dual pair for this representation in case $m-2n\not\in-2\mN$. Here we summarize the main results obtained in \cite{OSpHarm} on this Howe duality.
\begin{lemma}
\label{irrHk}
When $M=m-2n\not\in-2\mN$, the space $\cH_k$ of spherical harmonics on $\mR^{m|2n}$ of homogeneous degree $k$ is an irreducible $\mathfrak{osp}(m|2n)$-module. When $M\in-2\mN$, $\cH_k$ is irreducible if and only if 
\begin{eqnarray*}
k>2-M &\mbox{or}& k< 2-\frac{M}{2}.
\end{eqnarray*}
The module $\cH_k$ is always indecomposable. When reducible it has one submodule, $R^{2k+M-2}\cH_{2-M-k}$. When $\cH_k$ is irreducible, it is isomorphic to $K^{m|2n}_{k\epsilon_1}$ as an $\osp$-representation, otherwise the quotient with respect to the submodule is isomorphic to $K^{m|2n}_{k\epsilon_1}$.
\end{lemma}

This leads to the harmonic Fischer decomposition.

\begin{lemma}
If $M=m-2n \not \in -2 \mN$, $\cP$ decomposes into simple $\mathfrak{osp}(m|2n)$-modules as
\begin{eqnarray}
\cP = \bigoplus_{k=0}^{\infty} \cP_k= \bigoplus_{j=0}^{\infty} \bigoplus_{k=0}^{\infty} R^{2j}\cH_k.
\label{superFischer}
\end{eqnarray}
\label{superFischerLemma}
\end{lemma}
%\begin{proof}
%This is a combination of Lemma 2 and Theorem 13 in \cite{OSpHarm}.
%\end{proof}
Similar to the classical monogenic Fischer decomposition in Subsection \ref{secclassDi}, this decomposition implies that under the joint action of $\mathfrak{sl}_2\times\mathfrak{osp}(m|2n)$, the space $\cP\cong S\left( \mC^{m|2n}\right)$ is isomorphic to the multiplicity-free irreducible direct sum decomposition
\begin{eqnarray*}
\cP\cong \bigoplus_{k=0}^\infty T_{k+\frac{1}{2}M}\times K^{m|2n}_{k\epsilon_1},
\end{eqnarray*}
for $M\not\in-2\mN$, with $T_{\lambda}$ the irreducible $\mathfrak{sl}_2$-representation with lowest weight $\lambda$.

\subsection{Super spinor space}

The spinors $\mS_{m|2n}$ for $\osp$ are realizations of the Lie superalgebra as differential operators on the supersymmetric version of a Grassmann algebra, see \cite{Tensor} for the complete construction, characterization and motivation. This representation generalizes the spinor-representation for $\mathfrak{so}(m)$, but also corresponds to a notion of a minimal representation for $\osp$, similar to the metaplectic representation of $\mathfrak{sp}(2n)$, see \cite{Joseph}.

\begin{definition}
\label{superGrass}
\rm The complex algebra $\Lambda_{d|n}$ is freely generated by $\{\theta_1,\cdots,\theta_d,t_1,\cdots, t_n\}$ subject to the relations
\[\theta_{j}\theta_{k}=-\theta_k\theta_j\quad\mbox{for}\quad  1\le j,k\le d,\qquad t_it_l=t_lt_i\quad\mbox{for}\quad  1\le i,l\le n\]
and
\[\theta_jt_i=-t_i\theta_j\quad\mbox{for}\quad  1\le j\le d,\quad 1\le i\le n.\]
The parity which makes $\Lambda_{d|n}$ a superalgebra, is given by $|\theta_j|=0$ and $|t_i|=1$.
\end{definition}
The algebra $\Lambda_{d|n}$ is a super antisymmetric algebra, i.e. $ab=-(-1)^{|a||b|}ba$ for $a,b$ homogeneous elements of $\Lambda_{d|n}$.

%We also introduce the notation $T_j=\theta_j$ for $1\le j\le d$ and $T_{i+d}=t_i$ for $1\le i\le n$. The algebra $\Lambda_{d|n}$ is then generated by $T_i$ with commutation relations
%\begin{eqnarray*}
%T_kT_l&=&-(-1)^{|T_k| |T_l|}T_lT_k.
%\end{eqnarray*}
The subspaces of elements containing an even, respectively odd, amount of generators will be denoted by $\Lambda_{d|n}^+$ respectively $\Lambda_{d|n}^-$. These should not be confused with the even and odd part according to the $\mZ_2$-gradation, in which case the even part consists of elements containing an even amount of the odd generators.

The action of $\osp$ on $\Lambda_{d|n}$ for $m=2d+1$ and $m=2d$ will be given in Section \ref{superClifford}. This makes the algebra $\Lambda_{d|n}$ a simple $\mathfrak{osp}(2d+1|2n)$-module, denoted by 
\begin{eqnarray*}
\mS_{2d+1|2n}\cong K^{2d+1|2n}_{\omega_d-\frac{1}{2}\nu_n}\cong\Lambda_{d|n}.
\end{eqnarray*}
For the $\mathfrak{osp}(2d|2n)$-superalgebra, the module is the direct sum of two simple modules,
\begin{eqnarray*}
\mS_{2d|2n}=\mS^+_{2d|2n}\oplus\mS^-_{2d|2n}\cong K^{2d|2n}_{\omega_d-\frac{1}{2}\nu_n}\oplus K^{2d|2n}_{\omega_d+\nu_{n-1}-\frac{3}{2}\nu_n} \cong\Lambda_{d|n}^+\oplus\Lambda_{d|n}^-=\Lambda_{d|n}.
\end{eqnarray*}

In order to apply the Stein-Weiss procedure of \cite{MR0223492}, explained in Subsection \ref{secclassDi}, to the super spinor space, the decomposition into irreducible blocks of the tensor product with the fundamental representation is needed. This theorem and the subsequent Theorem \ref{decomp2} follow from Theorem 8 in \cite{Tensor}.
\begin{theorem}
\label{decomp1}
The tensor products of the fundamental representation of $\osp$ with the super spinor spaces satisfy
\begin{eqnarray*}
\mC^{2d+1|2n}\otimes\mS_{2d+1|2n}&\cong&\mC^{2d+1|2n}\boxtimes \mS_{2d+1|2n}\,\oplus\, \mS_{2d+1|2n}\\
\mC^{2d|2n}\otimes\mS_{2d|2n}^+\,&\cong&\mC^{2d|2n}\boxtimes \mS^+_{2d|2n}\,\,\,\oplus\,\, \mS_{2d|2n}^-\qquad\qquad\mbox{if}\quad d\not= n\\
\mC^{2d|2n}\otimes\mS_{2d|2n}^-\,&\cong&\mC^{2d|2n}\boxtimes \mS_{2d|2n}^-\,\,\,\oplus \,\,\mS_{2d|2n}^+\qquad\qquad\mbox{if}\quad d\not= n
\end{eqnarray*}
with $K^{m|2n}_\lambda\boxtimes K^{m|2n}_\mu = K^{m|2n}_{\lambda+\mu}$ the Cartan product. 

If $d=n$ the tensor products $\mC^{2n|2n}\otimes\mS_{2n|2n}^\pm$ are not completely reducible: there exist indecomposable highest weight representations $V^\pm$ such that
\begin{eqnarray*}
\mC^{2n|2n}\otimes\mS_{2n|2n}^\pm\supsetneq V^\pm\supsetneq \mS_{2n|2n}^\mp
\end{eqnarray*}
holds and the quotient $V^\pm/\mS_{2n|2n}^\mp$ is the irreducible highest weight representation with highest weight equal to the sum of the highest weights of $\mC^{2n|2n}$ and $\mS_{2n|2n}^\pm$.
\end{theorem}

In the rest of this paper, we will mostly consider superfunctions with values in the super spinor space. In particular, we will study the space $\cP\otimes\mS_{m|2n}$ and the properties of this function space as an $\osp$-module. Therefore, we look at the decomposition of the tensor product of the simple module $\cH_k$ with the super spinor space. According to Lemma \ref{irrHk}, the tensor product of the representation $K^{m|2n}_{k\epsilon_1}$ with spinor spaces needs to be studied.
\begin{theorem}
\label{decomp2}
The tensor product of the spherical harmonics on $\mR^{m|2n}$ of homogeneous degree $k$ or their simple quotient module, with the spinor spaces of $\osp$ decomposes into irreducible $\osp$-modules as follows: for $m=2d+1$
\[
\cH_k\otimes\mS_{2d+1|2n}\cong K^{2d+1|2n}_{k\epsilon_1}\otimes K^{2d+1|2n}_{\omega_d-\frac{1}{2}\nu_n}\cong K^{2d+1|2n}_{k\epsilon_1+\omega_d-\frac{1}{2}\nu_n}\oplus K^{2d+1|2n}_{(k-1)\epsilon_1+\omega_d-\frac{1}{2}\nu_n}
\]
holds and for $m=2d$ and $k\not=n-d+1$
\begin{eqnarray*}
K^{2d|2n}_{k\epsilon_1}\otimes\mS_{2d|2n}^+\cong K^{2d|2n}_{k\epsilon_1}\otimes K^{2d|2n}_{\omega_d-\frac{1}{2}\nu_n}&\cong&K^{2d|2n}_{k\epsilon_1+\omega_d-\frac{1}{2}\nu_n}\oplus K^{2d|2n}_{(k-1)\epsilon_1+\omega_d+\nu_{n-1}-\frac{3}{2}\nu_n}\\
K^{2d|2n}_{k\epsilon_1}\otimes\mS_{2d|2n}^-\cong K^{2d|2n}_{k\epsilon_1}\otimes K^{2d|2n}_{\omega_d+\nu_{n-1}-\frac{3}{2}\nu_n}&\cong&K^{2d|2n}_{k\epsilon_1+\omega_d+\nu_{n-1}-\frac{3}{2}\nu_n}\oplus K^{2d|2n}_{(k-1)\epsilon_1+\omega_d-\frac{1}{2}\nu_n}
\end{eqnarray*}
holds. If $k=n-d+1$ the tensor product is not completely reducible.
\end{theorem}
%If $d-n\le 0$, the full decomposition series of the tensor product $K^{2d|2n}_{(n-d+1)\epsilon_1}\otimes K^{2d|2n}_{\omega_d-\frac{1}{2}\nu_n}$ is given by
%\begin{eqnarray*}
%K^{2d|2n}_{(n-d+1)\epsilon_1}\otimes K^{2d|2n}_{\omega_d-\frac{1}{2}\nu_n}\supset\, V\,\supset K^{2d|2n}_{(n-d)\epsilon_1+\omega_d+\nu_{n-1}-\frac{3}{2}\nu_n}
%\end{eqnarray*}
%with $V$ an indecomposable highest weight module with highest weight $(n-d+1)\epsilon_1+\omega_d-\frac{1}{2}\nu_n$ and 
%\begin{eqnarray*}
%\left(K^{2d|2n}_{(n-d+1)\epsilon_1}\otimes K^{2d|2n}_{\omega_d-\frac{1}{2}\nu_n}\right)/ V&\cong &K^{2d|2n}_{(n-d)\epsilon_1+\omega_d+\nu_{n-1}-\frac{3}{2}\nu_n} \mbox{and}\\
% V/ K^{2d|2n}_{(n-d)\epsilon_1+\omega_d+\nu_{n-1}-\frac{3}{2}\nu_n}&\cong& K^{2d|2n}_{(n-d+1)\epsilon_1+\omega_d-\frac{1}{2}\nu_n}.
%\end{eqnarray*}
%The same holds for $K^{2d|2n}_{(n-d+1)\epsilon_1}\otimes K^{2d|2n}_{\omega_d+\nu_{n-1}-\frac{3}{2}\nu_n}$, but then $V$ has highest weight $(n-d+1)\epsilon_1+\omega_d+\nu_{n-1}-\frac{3}{2}\nu_n$ while the irreducible subrepresentation has highest weight $(n-d)\epsilon_1+\omega_d-\frac{1}{2}\nu_n$.

\section{The super Clifford algebra and $\osp$-spinors}
\label{superClifford}

As in the classical case the Dirac operator will be elegantly described in terms of a Clifford-type algebra. This super Clifford algebra will be identified with an algebra of endomorphisms on the super spinor space. The definition of the super Clifford algebra is an immediate graded extension of equation \eqref{classClifford}. We do not use the term Clifford-Weyl algebra here because then the generators of the Clifford algebra are usually assumed to commute with the the generators of the Weyl algebra, whereas in the following super Clifford algebra they mutually anticommute.

\begin{definition}
\label{defClifford}
{\rm The} super Clifford algebra $\sC$ \rm corresponding to the super vector space $\mC^{m|2n}$ is the associative superalgebra generated by $E_{k}$, where $k=1,\cdots,m+2n$. The gradation is given by $|E_k|=[k]$ and the multiplication relation is
\begin{eqnarray}
\label{commrelCliff}
E_kE_l+(-1)^{[k][l]}E_lE_k=-2g_{lk},
\end{eqnarray}
with $g$ the orthosymplectic metric \eqref{gmetric}.
\end{definition}
The super vector space $\mC^{m|2n}$ is naturally embedded in the super Clifford algebra by identifying the basis vectors $E_j$.

We define the superalgebra morphism $\widehat\cdot:a\to\widehat{a}$ on $\sC$ generated by
\begin{eqnarray*}
\widehat{E_i}&=&\sum_{j=1}^{m+2n}E_jg_{ji}\quad\mbox{ and}\\
\widehat{ab}&=&\widehat{a}\,\widehat{b}\mbox{ for }a,b\in\sC.
\end{eqnarray*}
This morphism is well-defined since the equation $\widehat{E_k}\widehat{E_l}+(-1)^{[k][l]}\widehat{E_l}\widehat{E_k}=-2g_{lk}$ holds.

\begin{definition}
\label{kappamap}
\rm The superalgebra morphism from the Clifford algebra $\cC l_{m|2n}$ to the endomorphism algebra End$(\mS_{m|2n})$ of spinor space $\mS_{m|2n}$ in Definition \ref{superGrass}, is given by $\kappa:$ $\sC\to$End$(\mS_{m|2n})$:
\begin{eqnarray*}
\kappa(E_j)=(\theta_j-\partial_{\theta_j}) & \mbox{for}&j=1,\cdots, d\\
\kappa(E_{d+j})=i(\theta_j+\partial_{\theta_j})& \mbox{for}&j=1,\cdots, d\\
\kappa(E_m)= i(-1)^{\sum_{j=1}^{d}\theta_j\partial_{\theta_j}+\sum_{i=1}^nt_i\partial_{t_i}}&\mbox{for} &\mbox{if }m=2d+1\\ 
\kappa(E_{m+i})=\sqrt{2}t_i&\mbox{for}&i=1,\cdots, n\\ 
\kappa(E_{m+n+i})=-\sqrt{2}\partial_{t_i}&\mbox{for}&i=1,\cdots ,n,
\end{eqnarray*}
with $d=\lfloor m/2\rfloor$.
\end{definition}
The operator $G=(-1)^{\sum_{j=1}^{d}\theta_j\partial_{\theta_j}+\sum_{i=1}^nt_i\partial_{t_i}}$ satisfies $\theta_iG=-G\theta_i$ and $t_iG=-Gt_i$. It is readily checked that $\kappa$ generates an algebra morphism by using the commutation relations in Definition \ref{superGrass}. When there is no confusion possible we will use the notation $A\cdot v$ for $\kappa(A)v$ with $A\in\sC$ and $v\in\mS_{m|2n}$.

The orthosymplectic Lie superalgebra $\osp$ can be embedded into the super Clifford algebra $\sC$ by an identification with the bi-vectors:
\begin{align}
\label{Bij}
\begin{split}
\iota:\cK_{kl}\rightarrow B_{kl}&=-\frac{1}{2}\left(g_{lk}+\sum_{a,b=1}^{m+2n}E_aE_bg_{ak}g_{bl}\right)\\
&=-\frac{1}{2}\left(\widehat{E}_k\widehat{E}_l+g_{lk}\right)\\
&=\frac{1}{4}\left((-1)^{[k][l]}\widehat{E_lE_k}-\widehat{E_kE_l}\right)
\end{split}
\end{align}

A direct calculation using Definition \ref{defClifford} shows that the corresponding bi-vectors $B_{kl}$ do satisfy the relations in formula \eqref{commL}.

The combination of the embedding of $\osp$ into $\sC$ with the action of $\sC$ on $\Lambda_{d|n}\cong\mS_{m|2n}$ in Definition \ref{kappamap} make $\mS_{m|2n}$ into an $\osp$-module, $\pi_{\mS}=\kappa\circ\iota$, or
\[
\pi_{\mS}(\cK_{kl})=\kappa(\iota(\cK_{kl})).
\]

\begin{remark}
\rm{The actions such as $\pi_{\cO}$ and $\pi_{\mS}$ will also be used to denote the action of $\osp$ on the tensor product of $\cO(\mR^{m|2n})$ or $\mS_{m|2n}$ with the trivial representation. So in particular for any $\osp$-module $V$, the formula $\pi_{\cO}(X)(f\otimes v)$ for $f\in\cO(\mR^{m|2n})$ and $v\in V$ is equal to $(\pi_{\cO}(X)f)\otimes v$, since $V$ is regarded here as a direct sum of trivial $\osp$-modules.}
\end{remark}

When we consider functions on $\mR^{m|2n}$ with values in a super vector space $V$, $\cO(\mR^{m|2n})\otimes V$, we can identify $f\otimes v$ with $fv$ or $(-1)^{|f||v|}vf$. It can be checked that the corresponding commutation relation $X_jv=(-1)^{[j]|v|}vX_j$ leads to a consistent $\mathfrak{osp}(m|2n)$-action. In the case of spinor valued functions $\cO(\mR^{m|2n})\otimes \mS_{m|2n}$, this implies that there is a gradation in the commutation of functions with the endomorphisms on $\mS_{m|2n}$. Therefore the commutation relation
\[
X_jE_k =(-1)^{[j][k]}E_k X_j
\]
holds.

\begin{remark}
\rm{In previous approaches to super Clifford analysis, see e.g. \cite{MR2374394}, a slightly different algebra was considered, where this commutation relation was $X_jE_k=E_k X_j$. However, this does not allow for an $\mathfrak{osp}(m|2n)$-symmetry of the basic operators.}
\end{remark}

Theorem \ref{decomp1} implies that there exists a surjective $\osp$-module morphism $E^\perp: \mC^{m|2n}\otimes \mS_{m|2n}\to \mS_{m|2n}$ if $m-2n\not=0$. For $m=2d+1$ this is unique (up to a multiplicative constant) while for $m=2d$ the space of such morphisms is two dimensional due to the fact that $\mS_{2d|2n}$ decomposes into two simple modules. For $m=2d$ however, there is also a preferred natural choice for $E^\perp$. The subsequent lemma also proves that such a vector space morphism still exists for the case $m-2n=0$.

\begin{lemma}
\label{Eperp}
The super vector space morphism $E^\perp: \mC^{m|2n}\otimes \mS_{m|2n}\to \mS_{m|2n}$ defined by the expression
\begin{eqnarray*}
E^\perp (E_k\otimes v)=\sum_{l=1}^{m+2n}(E_l\cdot v) g_{lk}=\kappa\left(\widehat{E_k}\right)v
\end{eqnarray*}
is an $\osp$-module morphism, thus it is invariant with respect to the $\osp$-action $\pi_{\mS}$ on the super spinor space and the natural action on $\mC^{m|2n}$ given in equation \eqref{nataction}.
\end{lemma}
\begin{proof}
It has to be proved that $E^\perp\circ\cK_{ij}=\cK_{ij}\circ E^\perp$, or 
\[
E^\perp\left((\cK_{ij}\cdot E_k)\otimes v\right)+(-1)^{[k]([i]+[j])}E^\perp(E_k\otimes \pi_\mS(\cK_{ij})v)=\pi_\mS(\cK_{ij})\kappa\left(\widehat{E_k}\right)v.
\]
Therefore we need the relation
%\begin{eqnarray*}
%[B_{ij},E_l]&=&(-1)^{[j]}\delta_{jl}\sum_{a=1}^{m+2n}E_a g_{ai}-(-1)^{[i][j]}(-1)^{[i]}\delta_{il}\sum_{b=1}^{m+2n}E_b g_{bj}
%\end{eqnarray*}%
%which implies
\begin{eqnarray}
\label{Baction}
[B_{ij},\widehat{E_k}]=-\frac{1}{2}\widehat{[E_iE_j,E_k]}=g_{kj}\widehat{E_i}-(-1)^{[i][j]}g_{ki}\widehat{E_j}\,=\,\widehat{\cK_{ij}\cdot E_k}.
\end{eqnarray}
Using this, we calculate
\begin{eqnarray*}
\cK_{ij}\cdot E^\perp (E_k\otimes v)&=&\kappa(\iota(\cK_{ij}))\kappa(\widehat{E_k}) v=\kappa(B_{ij}\widehat{E_k}) v) \\
&=&\kappa([B_{ij},\widehat{E_k}]) v+(-1)^{([i]+[j])[k]}(\widehat{E_k}\cdot(\cK_{ij}\cdot v)) \\
&=&\kappa(\widehat{\cK_{ij}\cdot E_k}) v+(-1)^{([i]+[j])[k]}(\widehat{E_k}\cdot(\cK_{ij}\cdot v)) \\
&=&E^\perp \left((\cK_{ij}\cdot E_k) \otimes v\right)+(-1)^{([i]+[j])[k]}E^\perp(E_k\otimes \cK_{ij}\cdot v)\\
&=&E^\perp \left(\cK_{ij}\cdot (E_k \otimes v)\right),
\end{eqnarray*}
which is the equation that needed to be proved.
\end{proof}

\section{The super Dirac operator}
\label{secdefDirac}

As in the classical case the super gradient is the $\osp$-invariant first order differential operator acting between $\cO(\mR^{m|2n})$ and $\cO(\mR^{m|2n})\otimes \mC^{m|2n}$, see the subsequent Lemma \ref{lemmagrad}. It is defined as
\begin{eqnarray}
\label{gradient}
\nabla f=\sum_{j=1}^{m+2n}(-1)^{[j](1+|f|)}\partial_{X_j}f\otimes E_j
\end{eqnarray}
for $f\in\cO(\mR^{m|2n})$ homogeneous. It is also naturally an $\mathfrak{osp}(m|2n)$-invariant operator between $\cO(\mR^{m|2n})\otimes V$ and $\cO(\mR^{m|2n})\otimes \mC^{m|2n}\otimes V$ for $V$ any $\osp$-representation.
\begin{lemma}
\label{lemmagrad}
The gradient $\nabla$ is an $\mathfrak{osp}(m|2n)$-invariant operator between $\cO(\mR^{m|2n})\otimes V$ and $\cO(\mR^{m|2n})\otimes \mC^{m|2n}\otimes V$ for $V$ any $\osp$-representation.
\end{lemma}
\begin{proof}
This follows immediately from the scalar case, where $V$ is equal to the trivial representation. The equation
\begin{equation}
\label{partL}
[\partial_{X_j},L_{kl}]=\delta_{jk}\partial_{X^l}-(-1)^{[k][l]}\delta_{jl}\partial_{X^k},
\end{equation}
which follows from the definition in equation \eqref{Lij} leads to the calculation
\begin{eqnarray*}
\nabla(L_{kl}f)&=&\sum_{j=1}^{m+2n}(-1)^{[j](1+|f|)}(L_{kl}\partial_{X_j}f)\otimes E_j\\
&+&(-1)^{[k](|f|+[l])}\partial_{X^l}f\otimes E_k-(-1)^{[k][l]}(-1)^{[l](|f|+[k])}\partial_{X^k}f\otimes E_l.
\end{eqnarray*}
This has to be equal to
\[
\cK_{kl}\cdot (\nabla f)=\sum_{j=1}^{m+2n}(-1)^{[j](1+|f|)}(L_{kl}\partial_{X_j}f)\otimes E_j+\sum_{j=1}^{m+2n}(-1)^{[j](1+|f|)}(-1)^{([j]+|f|)([k]+[l])}\partial_{X_j}f\otimes \cK_{kl}\cdot E_j.
\]
Using equation \eqref{nataction}, the second term is calculated as follows  
\begin{eqnarray*}
&&\sum_{j=1}^{m+2n}(-1)^{([j]+|f|)([j]+[k]+[l])}\partial_{X_j}f\otimes \cK_{kl}\cdot E_j\\
&=&\sum_{j=1}^{m+2n}(-1)^{([j]+|f|)[k]}\partial_{X_j}f\otimes g_{jl}E_k-\sum_{j=1}^{m+2n}(-1)^{([j]+|f|)[l]}(-1)^{[k][l]}\partial_{X_j}f\otimes g_{jk}E_l\\
&=&(-1)^{([l]+|f|)[k]}\partial_{X^l}f\otimes E_k-(-1)^{|f|[l]}\partial_{X^k}f\otimes E_l,
\end{eqnarray*}
which proves the lemma.
\end{proof}

Now we have all the necessary ingredients to define the super Dirac operator according to the Stein-Weiss procedure in \cite{MR0223492}. Only for the case $M=m-2n=0$ the tensor product in Theorem \ref{decomp1} of the fundamental representation with the spinorial representations does not decompose into a Cartan product part and a spinor representation, since this tensor product is not completely reducible. However, there is still a unique homomorphism $\mC^{2n|2n}\otimes\mS^\pm_{2n|2n}\to \mS^{\mp}_{2n|2n}$ and the corresponding $E^\perp$ is given in Lemma \ref{Eperp}. Thus, even though the notion of Cartan product is ill-defined for $\mC^{2n|2n}\otimes\mS^\pm_{2n|2n}$, the Stein-Weiss procedure remains applicable. %The appropriate projection $E^\perp$ in this case is given by the invariant projection onto the irreducible subrepresentation(s) with highest weight different from the sum of the two original highest weights. This will be further discussed in Section \ref{Fegan}.

\begin{definition}
\label{defDirac}
{\rm The} super Dirac operator \rm is the first order differential operator acting on the function space $\cO(\mR^{m|2n})\otimes\mS_{m|2n}$ given by
\begin{eqnarray*}
\px f= E^\perp (\nabla f)
\end{eqnarray*}
for $f\in\cO(\mR^{m|2n})\otimes\mS_{m|2n}$ with $E^\perp: \mC^{m|2n}\otimes \mS_{m|2n}\to \mS_{m|2n}$ given in Lemma \ref{Eperp} and $\nabla$ the gradient from equation \eqref{gradient}.
\end{definition}
As in the classical case the Dirac operator can be elegantly expressed by making use of the identification between the Clifford algebra and the endomorphisms on spinor space in Definition \ref{kappamap}.

\begin{proposition}
\label{DiracCliff}
The super Dirac operator of Definition \ref{defDirac} is given by
\begin{eqnarray*}
\px f= \sum_{j,k=1}^{m+2n}g_{jk}E_j\cdot(\partial_{X_k}f)
\end{eqnarray*}
for $f\in\cO(\mR^{m|2n})\otimes\mS_{m|2n}$. In this formula, $E_j \cdot h\otimes v$ should be interpreted as $(-1)^{[j]|h|}h\otimes E_j\cdot v=(-1)^{[j]|h|}h\otimes \kappa(E_j) v$ for $h\in\cO(\mR^{m|2n})$ homogeneous and $v\in\mS_{m|2n}$. The Dirac operator can equally be expressed as
\[
\px=\sum_{j,k=1}^{m+2n}\kappa(E_j)g_{jk}\partial_{X_k}=\sum_{k=1}^{m+2n}\kappa(\widehat{E_k})\partial_{X_k}=\sum_{j=1}^{m+2n}\partial_{X^j}\kappa(E_j).
\]
\end{proposition}

\begin{remark}
\label{remDiracCl}
\rm{The Dirac operator can now also be defined on super Clifford algebra valued functions, $\cO(\mR^{m|2n})\otimes\sC$ by the extension:
\begin{eqnarray*}
\px=\sum_{j,k=1}^{m+2n}E_jg_{jk}\partial_{X_k}=\sum_{k=1}^{m+2n}\widehat{E_k}\partial_{X_k}=\sum_{j=1}^{m+2n}\partial_{X^j}E_j,
\end{eqnarray*}
where $E_j$ acts by left multiplication on $\cO(\mR^{m|2n})\otimes\sC$.}
\end{remark}

The invariance of the gradient $\nabla: \cO(\mR^{m|2n})\otimes \mS_{m|2n}\to \cO(\mR^{m|2n})\otimes\mC^{m|2n}\otimes \mS_{m|2n}$ and the projection $E^\perp:\cO(\mR^{m|2n})\otimes\mC^{m|2n}\otimes \mS_{m|2n}\to \cO(\mR^{m|2n})\otimes \mS_{m|2n}$ imply that the Dirac operator
\begin{eqnarray*}
\px:\cO(\mR^{m|2n})\otimes \mS_{m|2n}\to\cO(\mR^{m|2n})\otimes \mS_{m|2n}
\end{eqnarray*}
is $\mathfrak{osp}(m|2n)$-invariant. This can be expressed as
\begin{eqnarray}
\label{ospDirac}
 \px K_{ij}=K_{ij}\px\qquad\mbox{for} \quad K_{ij}=L_{ij}+B_{ij}
 \end{eqnarray}
where $L_{ij}$ and $B_{ij}$ are defined in equation \eqref{Lij} and \eqref{Bij} respectively. It follows immediately from the fact that $\cO(\mR^{m|2n})\otimes\mS_{m|2n}$ is a tensor product representation that the operators $K_{ij}$ generate $\osp$. This can also be calculated directly by taking into account the correct commutation relation $E_jX_k=(-1)^{[j][k]}X_kE_j$. This also implies that the Dirac operator acting on super Clifford algebra functions in Remark \ref{remDiracCl} is $\mathfrak{osp}(m|2n)$-invariant.

\begin{remark}
\rm{In the present section the condition $m>0$ is never needed, since the spinors for $\mathfrak{sp}(2n)$ behave as the limit $m\to 0$ for $\mathfrak{osp}(m|2n)$. So also the Dirac operator on $\mR^{0|2n}$ can be constructed using the Stein-Weiss procedure. In \cite{Krysl} this operator was introduced in a different context and its Howe duality was also studied. So in the current paper we will not focus on this particular limit case.}
\end{remark}

%\section{Conformal Killing vector fields on $\mR^{m|2n}$}

Now we show how the Dirac operator on $\mR^{m|2n}$ is closely related with an operator in the construction of Bernstein, which has been generalized by Shmelev, see \cite{MR1893457, Shmelev2}. On superspace $\mR^{2n|m}$ with coordinates $(y_1,\cdots,y_{2n},{y\grave{}}_1,\cdots,{y\grave{}}_{m})$ we can consider differential forms, see e.g. the basic definitions in \cite{MR2374394}. Then the operators
\begin{eqnarray*}
X_+&=&2\sum_{i=1}^ndy_idy_{i+n}+\sum_{j=1}^m\left(d{y\grave{}}_{j}\right)^2\mbox{ and}\\
X_-&=&2\sum_{i=1}^n\partial_{y_i}\rfloor\partial_{y_{i+n}}\rfloor+\sum_{j=1}^m\left(\partial_{{y\grave{}}_{j}}\rfloor\right)^2
\end{eqnarray*}
generate $\mathfrak{sl}_2$. This is equivalent with the fact that the operators in equation \eqref{defsuperLap} generate $\mathfrak{sl}_2$. Then the operator 
\begin{eqnarray*}
D_-&=&\sum_{i=1}^n\left(\partial_{y_i}\partial_{y_{i+n}}\rfloor-\partial_{y_{i+n}}\partial_{y_{i}}\rfloor\right)-\frac{1}{2}\sum_{k=1}^{2n}y_k\partial_{y_k}\rfloor+\sum_{i=1}^m\left(\partial_{{y\grave{}}_{j}}+{y\grave{}}_{j}\right)\partial_{{y\grave{}}_{j}}\rfloor
\end{eqnarray*}
satisfies $D_-^2=X_-$ and together with $X_+$ and $X_-$ this operator generates $\mathfrak{osp}(1|2)$. The operator $D_+$ has a similar expression, but is also defined by $[D_-,X_+]$. 

Now we make an identification $dy_k\leftrightarrow {x\grave{}}_k$, $d{y\grave{}}_{j}\leftrightarrow x_j$, $\partial_{y_k}\rfloor\leftrightarrow \partial_{{x\grave{}}_k}$, $\partial_{{y\grave{}}_{j}}\rfloor\leftrightarrow \partial_{x_j}$ with $(\ux,\uxb)$ the coordinates on $\mR^{m|2n}$, then
\begin{eqnarray*}
X_+=R^2,&&X_-=\Delta\qquad\mbox{and}\\
D_-&=&\sum_{i=1}^n\left(-\partial_{y_{i+n}}-\frac{1}{2}y_i\right)\partial_{{x\grave{}}_i}+\sum_{i=1}^n\left(\partial_{y_i}-\frac{1}{2}y_{i+n}\right)\partial_{{x\grave{}}_{i+n}}+\sum_{i=1}^m\left(\partial_{{y\grave{}}_{j}}+{y\grave{}}_{j}\right)\partial_{x_j}.
\end{eqnarray*}
So this constitutes a different square root of the super Laplace operator leading to a realization of $\mathfrak{osp}(1|2)$. In this construction a Clifford-Weyl algebra with $2n$ bosonic variables and $m$ fermionic variables is used, while the Dirac operator in Proposition \ref{DiracCliff} is constructed using a super Clifford algebra with only $n$ bosonic variables and $\lfloor m/2\rfloor$ fermionic variables, see the identification in Definition \ref{kappamap}. The Howe duality in Bernsteins construction is on the space of differential forms, the dual partner of $\mathfrak{osp}(1|2)$ being the super Poisson algebra $\mathfrak{po}(2n|m)$, which contains $\mathfrak{osp}(m|2n)$. In the following section we will study the Howe duality in our construction for $\mathfrak{osp}(1|2)\times\osp$.

\section{Howe duality and Fischer decomposition}
\label{secHowe}
As in the classical case, the square of the super Dirac operator is given by minus the super Laplace operator \eqref{defsuperLap}, $\px^2=-\Delta$ which is a natural requirement for a proper super Dirac operator. This and other commutation relations are calculated in Theorem \ref{osp12} which shows that the super Dirac operator, together with the vector variable $\bold{x}$ defined below, generates the Lie superalgebra $\mathfrak{osp}(1|2)$. This realization of $\mathfrak{osp}(1|2)$ commutes with the action of $\osp$ on $\cO(\mR^{m|2n})\otimes \mS_{m|2n}$. In this section we prove that these two algebras constitute a Howe dual pair if $m-2n\not\in-2\mN$. We also determine the action of the Lie superalgebra $\mathfrak{osp}(m+4n|2m+2n)$, in which $\mathfrak{osp}(1|2)$ and $\osp$ are each other's centralizers, on $\cO(\mR^{m|2n})\otimes \mS_{m|2n}$.

In the original setting in \cite{MR0986027}, dual pairs were considered where one of the two algebras is equal to $\mathfrak{gl}(k)$, $\mathfrak{so}(m)$ or $\mathfrak{sp}(2n)$. Consider $\mk$ equal to $\mathfrak{so}(m)$ or $\mathfrak{sp}(2n)$ and $V$ its natural representation, $\mC^{m}$ or $\mC^{2n}$ respectively. In \cite{MR0986027} the representation structure of $\mk$ on $\left(\otimes^kS(V)\right)\otimes \left(\otimes^l\Lambda(V)\right)$ is investigated in the context of invariant theory. There is a natural action of $\mathfrak{osp}(2ld|2kd)$ (with $d=\dim V$, which is $m$ or $2n$) on this space, as in \cite{OSpHarm, Tensor, MR0986027, MR1893457, Kyo}. This action of $\mathfrak{osp}(2ld|2kd)$ includes the action of $\mk$. The representation of $\mk$ is then studied by considering its centraliser in $\mathfrak{osp}(2ld|2kd)$. One example is $\mk=\mathfrak{so}(m)$, $k=1$ and $l=0$, then the centraliser of $\mathfrak{so}(m)$ inside $\mathfrak{sp}(2m)$ is $\mathfrak{sl}(2)$. This is the Howe duality corresponding to the Laplace operator on $\mR^m$.

There are several extensions of this principle, which are also known as Howe dualities and which possess similar properties as the original dual pairs in \cite{MR0986027}. We mention two concrete examples, which are extensions of the specific Howe duality given above, for others see e.g. \cite{MR2646304, MR1847665, MR1893457}. The first one is studied in \cite{OSpHarm} and is a superization of the example above. The algebra $\mk$ becomes the Lie superalgebra $\mathfrak{osp}(m|2n)$, the role of $V$ replaced by its natural representation $\mC^{m|2n}$ and the again the choice $k=1$, $l=0$ is made. Since $S(\mC^{m|2n})\cong S(\mC^{m})\otimes\Lambda(\mC^{2n})$ there is a natural action of $\mathfrak{osp}(4n|2m)$, in which $\mathfrak{osp}(m|2n)$ is included and has centraliser $\mathfrak{sl}(2)$. This is the Howe duality corresponding to the super Laplace operator on $\mR^{m|2n}$. In \cite{OSpHarm} it is proven that whenever $m-2n\not\in-2\mN$ the dual pair $(\mathfrak{osp}(m|2n),\mathfrak{sl}(2))$ possesses the appropriate properties for a Howe dual pair.

Another extension considers $\mk$ still equal to $\mathfrak{so}(m)$ (for convenience we choose $m$ even), but now the action is studied on $S(\mC^m)\otimes\Lambda(\mC^{\frac{m}{2}})$ where $\Lambda(\mC^{\frac{m}{2}})$ has the $\mathfrak{so}(m)$ representation structure of the spinor space $\mS_m$ as considered in Subsection \ref{secclassDi}. Now there is a natural action of $\mathfrak{osp}(m|2m)$ and the centraliser of $\mathfrak{so}(m)$ in this algebra is given by $\mathfrak{osp}(1|2)$. This Howe duality corresponds to the Dirac operator and is studied in \cite{Kyo}, where it is proven that this does indeed satisfy appropriate properties to be called a Howe dual pair.

The super Dirac operator in this paper provides a way of combining these two extensions of the principle of Howe dualities, which will be explored in this section.

The vector variable is defined as an element of $\cO(\mR^{m|2n})\otimes \sC$,
\[
\bold{x}=\sum_{j=1}^{m+2n}X_jE_j,
\]
which can be seen as an operator on $\cO(\mR^{m|2n})\otimes \mS_{m|2n}$ via $\bold{x}=\sum_{j}X_j\kappa(E_j)$ with the $\kappa$-action in Definition \ref{kappamap} or as an operator on $\cO(\mR^{m|2n})\otimes \sC$ by left multiplication.

It can be easily checked that, as the Dirac operator in equation \eqref{ospDirac}, the operator $\bold{x}$ is $\mathfrak{osp}(m|2n)$-invariant. %A useful property of the vector variable is 
%\begin{eqnarray}
%\label{partialX}
%\partial_{X^j}\bold{x}&=&\widehat{E}_j.
%\end{eqnarray}

\begin{theorem}
\label{osp12}
The odd operators $\px$ and $\bold{x}$ acting on $\cO(\mR^{m|2n})\otimes\mS_{m|2n}$ or $\cO(\mR^{m|2n})\otimes\sC$ generate the Lie superalgebra $\mathfrak{osp}(1|2)$.
\end{theorem}

\begin{proof}
Definition \ref{defClifford} allows to calculate
\begin{eqnarray*}
[\bold{x},\bold{x}]&=&2\bold{x}^2=\sum_{j,k}(-1)^{[j][k]}X_jX_k E_jE_k+\sum_{j,k}X_jX_k E_kE_j\\
&=&-2\sum_{j,k}X_jX_k g_{jk}= -2 R^2.
\end{eqnarray*}
Similarly, the expression for the Dirac operator using the Clifford algebra in Proposition \ref{DiracCliff} or Remark \ref{remDiracCl} allows to compute that $[\px,\px]=2\px^2= -2 \Delta$.

Then we prove that the equality $[\px,\bold{x}] =- 2 (\mE + \frac{1}{2}M)$ holds by calculating
\begin{eqnarray*}
\px\bold{x}+\bold{x}\px&=&\sum_{l,k=1}^{m+2n}E_lg_{lk}E_k+\sum_{j,l,k=1}^{m+2n}X_j\left((-1)^{[j][l]}E_lE_j+E_jE_l\right)g_{lk}\partial_{X_k}\\
&=&\frac{1}{2}\sum_{l,k=1}^{m+2n}\left(E_lg_{lk}E_k +(-1)^{[k][l]}E_kg_{lk}E_l\right)-2\sum_{j,l,k=1}^{m+2n}X_jg_{lj}g_{lk}\partial_{X_k}\\
&=&-M-2\sum_{k=1}^{m+2n}X_k\partial_{X_k}.\\
\end{eqnarray*}
It is already known that $R^2$, $\Delta$ and $\mE + \frac{1}{2}M$ generate the Lie algebra $\mathfrak{sl}_2\cong\mathfrak{sp}(2)$, the underlying Lie algebra of $\mathfrak{osp}(1|2)$. The mixed commutators are given by
\[
\begin{array}{lll}
\left[\bold{x}, \bold{x}^2\right] =0 &\qquad& \left[\px, \bold{x}^2\right] =-2\bold{x}\\
\left[\bold{x}, \Delta \right] =-2\px &\qquad& \left[\px, \Delta\right] =0\\
\left[\bold{x}, \mE + \frac{1}{2}M\right] =-\bold{x} &\qquad& \left[\px, \mE + \frac{1}{2}M \right] =\px,
\end{array}
\]
which concludes the proof.
\end{proof}

\begin{remark}
\rm{The Dirac operator $\px:\cP\otimes \mS_{m|2n}\to\cP\otimes \mS_{m|2n}$ or $\px:\cP\otimes\sC\to\cP\otimes\sC$ is surjective if $m\not=0$. This follows immediately from the fact that the square $\px^2=-\Delta$ is surjective on $\cP$, which is a consequence of the surjectivity of the classical Laplace operator. }
\end{remark}

The null-solutions of the super Dirac operator are called {\em super monogenic functions}. The space of {\em spherical monogenics} of degree $k$ is given by
\begin{eqnarray*}
\cM_k=\{p\in \cP\otimes \mS_{m|2n}|\,\mE p=k p\mbox{ and }\px p=0\}.
\end{eqnarray*}
The spaces $\cM_k^\pm$ for $m=2d$ are the monogenic functions of degree $k$ with values in $\mS^\pm_{2d|2n}$.

The relation $[\px,\bold{x}] =- 2 \mE - M$ implies that the equalities
\begin{eqnarray}
\label{px2l}
\px \bold{x}^{2l}M_k&=&-2l \bold{x}^{2l-1}M_k\\
\label{px2l1}
\px \bold{x}^{2l+1}M_k&=&-(2k+2l+M)\bold{x}^{2l}M_k
\end{eqnarray}
hold for $M_k\in\cM_k$.

\begin{theorem}
\label{HkS}
For $\cH_k$, the space of spherical harmonics of degree $k$ on $\mR^{m|2n}$, the decomposition
\begin{eqnarray*}
\cH_k\otimes\mS_{m|2n}=\cM_k\otimes\bold{x}\cM_{k-1}
\end{eqnarray*}
holds if $k\not=1-\frac{1}{2}M$. If $m$ is odd this is a decomposition into simple $\osp$-modules. If $m$ is even (with $k\not=1-\frac{1}{2}M$) this can be refined to
\begin{eqnarray*}
\cH_k\otimes\mS_{m|2n}^+&=&\cM_k^+\otimes\bold{x}\cM_{k-1}^-\\
\cH_k\otimes\mS_{m|2n}^-&=&\cM_k^-\otimes\bold{x}\cM_{k-1}^+
\end{eqnarray*}
and if $k<1-\frac{1}{2}M$ or $k>2-M$ holds these correspond to decompositions into simple $\osp$-modules.
\end{theorem}
\begin{proof}
It is clear that $\px\px M_k=0$ for $M_k\in\cM_k$ and as a consequence of equation \eqref{px2l1} we obtain that also $\px\px \bold{x}M_{k-1}=0$ holds for $M_{k-1}\in\cM_{k-1}$. Equation \eqref{px2l1} also implies that $\cM_k\cap\bold{x}\cM_{k-1}=0$ if $k\not=1-\frac{1}{2}M$. Now if $k\not=1-\frac{1}{2}M$, every harmonic $H_k\in\cH_k\otimes\mS_{m|2n}$ can be written as
\[
H_k=\left(H_k+\frac{1}{2k-2+M}\bold{x}\px H_k\right)-\frac{1}{2k-2+M}\bold{x}\px H_k,
\]
which proves $\cH_k\otimes\mS_{m|2n}=\cM_k\otimes\bold{x}\cM_{k-1}$, by applying equation \eqref{px2l1}.

Since $\px$ and $\mE$ commute with the $\osp$-action on $\cO(\mR^{m|2n})\otimes\mS_{m|2n}$, the spaces $\cM_k$ are $\osp$-modules. If $m$ is odd $k\not=1-\frac{1}{2}M$ always holds. Comparison with Lemma \ref{irrHk} and Theorem \ref{decomp2} then shows that for $m=2d+1$, $\cM_k\cong K^{2d+1|2n}_{k\epsilon_1+\omega_d-\frac{1}{2}\nu_n}$ holds and the theorem follows.

The case $m=2d$ is proven similarly.
\end{proof}

The proof of Theorem \ref{HkS} implies the following corollary.

\begin{corollary}
\label{monosp}
The spherical monogenics on $\mR^{2d+1|2n}$ satisfy
\begin{eqnarray*}
\cM_k\cong K^{2d+1|2n}_{k\epsilon_1+\omega_d-\frac{1}{2}\nu_n}
\end{eqnarray*}
as $\mathfrak{osp}(2d+1|2n)$-representations. The spherical monogenics on $\mR^{2d|2n}$ satisfy
\begin{eqnarray*}
\cM_k^+\cong K^{2d|2n}_{k\epsilon_1+\omega_{d}-\frac{1}{2}\nu_n}&\mbox{and }&\cM_k^-\cong K^{2d|2n}_{k\epsilon_1+\omega_{d}+\nu_{n-1}-\frac{3}{2}\nu_n}
\end{eqnarray*}
if $d>n$ or $d\le n$ with $k\not\in[1+n-d,1+2n-2d]$.
\end{corollary}
\begin{proof}
This is a consequence of the results in Theorem \ref{decomp2} and Theorem \ref{HkS}.
\end{proof}
The remaining cases of $\cM_k^{\pm}$ will be dealt with in Section \ref{Mkasreps}. 

\begin{remark}
\rm{In \cite{Lavicka} and \cite{Lavicka2} Gelfand-Tsetlin bases for the spaces of monogenic polynomials on $\mR^m$ are constructed. It is an interesting question whether these methods extend to the monogenics on $\mR^{m|2n}$.}
\end{remark}

In case $M\not\in-2\mN$, the previous results can be used to obtain the Howe duality for super Clifford analysis.

\begin{theorem}
The monogenic Fischer decomposition on $\mR^{m|2n}$ for $m-2n\not\in-2\mN$, is given by
\begin{eqnarray*}
\cP\otimes \mS_{m|2n}=\bigoplus_{j=0}^\infty\bigoplus_{k=0}^\infty \bold{x}^j\cM_k.
\end{eqnarray*}
For $m=2d+1$, this is a decomposition into simple $\mathfrak{osp}(2d+1|2n)$-modules.

If $m=2d$, with $d-n>0$, the decomposition into simple $\mathfrak{osp}(2d|2n)$-modules is given by
\[
\cP\otimes \mS_{2d|2n}=\left(\bigoplus_{j=0}^\infty\bigoplus_{k=0}^\infty \bold{x}^j\cM^{+}_k\right)\bigoplus \left(\bigoplus_{j=0}^\infty\bigoplus_{k=0}^\infty \bold{x}^j\cM^{-}_k\right).
\]
\end{theorem}
\begin{proof}
This is a combination of the scalar Fischer decomposition in Lemma \ref{superFischerLemma} and the result in Theorem \ref{HkS}.
\end{proof}

The combination of this theorem with equations \eqref{px2l} and \eqref{px2l1} implies that the spinor-valued polynomials on $\mR^{2d+1|2n}$ are isomorphic to an irreducible multiplicity-free direct sum decomposition under the joint action of $\mathfrak{osp}(1|2)\times\osp$. This is given by
\[
\cP\otimes\mS_{2d+1|2n}\cong\bigoplus_{k=0}^\infty \left(T^{1|2}_{k+d-n+\frac{1}{2}}\times K^{2d+1|2n}_{k\epsilon_1+\omega_d-\frac{1}{2}\nu_n}\right),
\]
with $T^{1|2}_{\alpha}$ the irreducible $\mathfrak{osp}(1|2)$-representation with lowest weight $\alpha$.

If $d>n$, the multiplicity-free irreducible direct sum decompositions on $\mR^{2d|2n}$ is given by
\[
\cP\otimes\mS_{2d|2n}\cong \bigoplus_{k=0}^\infty \left(T^{1|2}_{k+d-n}\times \left(K^{2d|2n}_{k\epsilon_1+\omega_d-\frac{1}{2}\nu_n}\oplus  K^{2d|2n}_{k\epsilon_1+\omega_d+\nu_{n-1}-\frac{3}{2}\nu_n}\right)\right).
\]
This does not yet correspond to a true Howe duality, because each different representation of $\mathfrak{osp}(1|2)$ corresponds to two irreducible $\mathfrak{osp}(m|2n)$-representations. As we will see later, the Howe duality splits into two parts:
\begin{eqnarray*}
\cP^+\otimes\mS_{2d|2n}^+\,\oplus\, \cP^-\otimes\mS_{2d|2n}^-&\cong&\bigoplus_{k=0}^\infty \left(T^{1|2}_{k+d-n}\times K^{2d|2n}_{k\epsilon_1+\omega_d-\frac{1}{2}\nu_n-\frac{1}{2}(1-(-1)^k)\delta_n}\right)\quad\mbox{and}\\
\cP^+\otimes\mS_{2d|2n}^+\,\oplus \,\cP^-\otimes\mS_{2d|2n}^-&\cong&\bigoplus_{k=0}^\infty \left(T^{1|2}_{k+d-n}\times K^{2d|2n}_{k\epsilon_1+\omega_d-\frac{1}{2}\nu_n-\frac{1}{2}(1+(-1)^k)\delta_n}\right),
\end{eqnarray*}
with $\cP^\pm$ the spaces of polynomials consisting of an even, respectively odd, amount of variables.

In order to discuss the Howe duality further we work with complex Lie superalgebras in the remainder of this section. As can be seen from the table in \cite{MR1893457}, the Lie superalgebras $\mathfrak{osp}(1|2;\mC)$ and $\mathfrak{osp}(m|2n;\mC)$ are each others centralizers inside the Lie superalgebra $\mathfrak{osp}(m+4n|2m+2n;\mC)$. Now we show how we can realize $\mathfrak{osp}(m+4n|2m+2n;\mC)$ on the representation space $\cP\otimes\mS_{m|2n}$, which completes the study of the Howe duality. The classical case of this Howe duality is $\mathfrak{osp}(1|2)\times\mathfrak{so}(m)\subset\mathfrak{osp}(m|2m)$, see \cite{Kyo}. The other classical limit is $\mathfrak{osp}(1|2)\times\mathfrak{sp}(2n)\subset\mathfrak{osp}(4n|2n)$, which has been studied in \cite{Krysl}.%The Howe duality corresponding to the Laplace operator on $\mR^m$ is given by $\mathfrak{sl}_2\times\mathfrak{so}(m)\subset \mathfrak{sp}(2m)$, see \cite{MR0986027}. This has also been generalized to $\mR^{m|2n}$ i
  n \cite
 {OSpHarm} as $\mathfrak{sl}_2\times\mathfrak{osp}(m|2n)\subset\mathfrak{osp}(4n|2m)$.

\begin{definition}
\label{defbigosp}
The Lie superalgebra $\mathfrak{g}$ \rm is generated by the operators
\begin{eqnarray*}
X_iX_j(-1)^{([i]+[j])\mE},\qquad 2X_i\partial_{X^j}(-1)^{([i]+[j])\mE}+g_{ji},\qquad \partial_{X^i}\partial_{X^j}(-1)^{([i]+[j])\mE},\qquad B_{ij}(-1)^{([i]+[j])\mE}
\end{eqnarray*}
for $i,j=1,\cdots,m+2n$ and with $B_{ij}$ the bi-vectors from equation \eqref{Bij} and by the operators
\begin{eqnarray*}
X_iE_j(-1)^{([i]+[j])\mE},\qquad E_j\partial_{X^i}(-1)^{([i]+[j])\mE}\qquad\qquad\mbox{for }i,j=1,\cdots,m+2n.
\end{eqnarray*}
The gradation on $\mg$ is induced by $|X_j|=(-1)^{[j]}$ and $|E_j|=1-(-1)^{[j]}$. In fact the choice $|X_j|=1-(-1)^{[j]}$ and $|E_j|=(-1)^{[j]}$ would have the same resulting gradation on $\mathfrak{g}$. 
\end{definition}

Obviously the realization of the Lie superalgebra $\mathfrak{osp}(1|2)$ on $\cP\otimes\mS_{m|2n}$ is embedded in this algebra $\mathfrak{g}$. The operators $K_{ij}(-1)^{([i]+[j])\mE}$ with $K_{ij}$ given in equation \eqref{ospDirac} are also inside the algebra $\mg$. Since $\mE$ commutes with $K_{ij}$ these operators still satisfy commutation relation \eqref{commL}. This realization of $\osp$ on $\cP\otimes\mS_{m|2n}$ clearly has the exact same properties so we identify this realization with the previous one studied in the current paper.

The algebra $\mg$ is defined by operators on $\cP\otimes\mS_{m|2n}$ and hence $\cP\otimes\mS_{m|2n}$ is immediately a $\mg$-module.

\begin{theorem}
The Lie superalgebra $\mg$ from Definition \ref{defbigosp} is isomorphic to the Lie superalgebra $\mathfrak{osp}(m+4n|2m+2n)$ and the representation on $\cP\otimes\mS_{m|2n}$ is irreducible if $m=2d+1$, 
\[\cP\otimes\mS_{2d+1|2n}\cong K^{2d+1+4n|4d+2+2n}_{\omega_{d+2n}-\frac{1}{2}\nu_{n+2d+1}}\]
while for $m=2d$ it decomposes into two irreducible modules as
\begin{eqnarray*}
\cP\otimes \mS_{2d|2n}&=&\left((\cP^+\otimes\mS_{2d|2n}^+)\oplus(\cP^-\otimes\mS_{2d|2n}^-)\right)\bigoplus\left((\cP^+\otimes\mS_{2d|2n}^-)\oplus(\cP^-\otimes\mS_{2d|2n}^+)\right)\\
&\cong&\,K^{2d+4n|4d+2n}_{\omega_{d+2n}-\frac{1}{2}\nu_{n+2d}}\quad \bigoplus\quad K^{2d+4n|4d+2n}_{\omega_{d+2n}+\nu_{n+2d-1}-\frac{3}{2}\nu_{n+2d}}.
\end{eqnarray*}
\end{theorem}
Before we prove this theorem we note that this different behavior for $m$ even or odd corresponds to the observed properties for the multiplicity-free irreducible direct sum decompositions under $\mathfrak{osp}(1|2)\times\osp$ earlier in this section. Each simple $\mg$-submodule of $\cP\otimes \mS_{2d|2n}$ leads to a realization of the Howe duality $\mathfrak{osp}(1|2)\times \mathfrak{osp}(m|2n)$.
\begin{proof}
The irreducibility and decomposition of $\cP\otimes\mS_{m|2n}$ as a $\mg$-representation follow immediately from the definition of $\mg$ and Definition \ref{kappamap}.

To prove the claim $\mg\cong\mathfrak{osp}(m+4n|2m+2n)$ we restrict to $m=2d$, with the case $m=2d+1$ being similar. First of all it is easier to replace the Clifford algebra  elements $E_j$, $j=1,\cdots,2d$ by their corresponding Grassmann variables $\theta_i,\partial_{\theta_i}$, $i=1,\cdots,d$ from Definition \ref{kappamap} and likewise we express $E_j$, $j=2d+1,\cdots,2d+2n$ in terms of $t_k,\partial_{t_k}$, $k=1,\cdots,n$. The variables $X_j(-1)^{[j]\mE}$, $t_i$ and $\theta_k(-1)^\mE$ correspond to $2d+n$ commuting variables and $2n+d$ anti-commuting variables such that the commuting and anti-commuting variables mutually anti-commute. The operators in the definition of $\mg$ are then exactly the quadratic elements in the algebra generated by these variables and their partial derivatives which generate the Lie superalgebra $\mathfrak{osp}(2d+4n|4d+2n)$, as in the oscillator realization or super spinor realizations of orthosymplectic superalgebras, see e.g. \cite{OSpHarm, 
 Tensor, Kyo}. The highest weight then follows immediately from \cite{Tensor}.
\end{proof}

The realizations of the Howe duality are thus given by
\begin{equation}
\label{sumHowe1}
K^{2d+1+4n|4d+2+2n}_{\omega_{d+2n}-\frac{1}{2}\nu_{n+2d+1}}\cong \bigoplus_{k=0}^\infty \left(T^{1|2}_{k+d-n+\frac{1}{2}}\times K^{2d+1|2n}_{k\epsilon_1+\omega_d-\frac{1}{2}\nu_n}\right),
\end{equation}
for $\mathfrak{osp}(2d+1|2n)\times\mathfrak{osp}(1|2)\subset\mathfrak{osp}(2d+1+4n|4d+2+2n)$ and
\begin{eqnarray}
\label{sumHowe2}
K^{2d+4n|4d+2n}_{\omega_{d+2n}-\frac{1}{2}\nu_{n+2d}}&\cong&\bigoplus_{k=0}^\infty \left(T^{1|2}_{k+d-n}\times K^{2d|2n}_{k\epsilon_1+\omega_d-\frac{1}{2}\nu_n-\frac{1}{2}(1-(-1)^k)\delta_n}\right)\quad\mbox{and}\\
\label{sumHowe3}
K^{2d+4n|4d+2n}_{\omega_{d+2n}+\nu_{n+2d-1}-\frac{3}{2}\nu_{n+2d}}&\cong&\bigoplus_{k=0}^\infty \left(T^{1|2}_{k+d-n}\times K^{2d|2n}_{k\epsilon_1+\omega_d-\frac{1}{2}\nu_n-\frac{1}{2}(1+(-1)^k)\delta_n}\right),
\end{eqnarray}
for $\mathfrak{osp}(2d|2n)\times\mathfrak{osp}(1|2)\subset\mathfrak{osp}(2d+4n|4d+2n)$ if $d>n$.

For the sake of completeness, we repeat this main result in the distinguished root system, see \cite{MR051963}. We denote a highest weight representation of $\osp$ with highest weight $\mu$ with respect to that root system by $L^{m|2n}_\mu$. The conversion from one root system to the other is summarized in Section 4 in \cite{Tensor}. Explicitly we have the following:
\begin{eqnarray*}
L^{2d+1+4n|4d+2+2n}_{\omega_{d+2n}-\frac{1}{2}\nu_{n+2d+1}}&\cong&\bigoplus_{k=0}^n \left(T^{1|2}_{k+d-n+\frac{1}{2}}\times L^{2d+1|2n}_{\omega_d+\nu_k-\frac{1}{2}\nu_n}\right) \,\oplus\,\bigoplus_{k=n+1}^\infty \left(T^{1|2}_{k+d-n+\frac{1}{2}}\times L^{2d+1|2n}_{(k-n)\epsilon_1+\omega_d+\frac{1}{2}\nu_n}\right),\\
L^{2d+4n|4d+2n}_{\omega_{d+2n}-\frac{1}{2}\nu_{n+2d}}&\cong&\bigoplus_{k=0}^n \left(T^{1|2}_{k+d-n}\times L^{2d|2n}_{\omega_{d-\frac{1}{2}(1-(-1)^k)}+\nu_k-\frac{1}{2}\nu_n}\right)\\
&&\oplus\bigoplus_{k=n+1}^\infty \left(T^{1|2}_{k+d-n}\times L^{2d|2n}_{(k-n)\epsilon_1+\omega_{d-\frac{1}{2}(1-(-1)^k)}+\frac{1}{2}\nu_n}\right)\quad\mbox{and}\\
L^{2d+4n|4d+2n}_{\omega_{d+2n-1}-\frac{1}{2}\nu_{n+2d}}&\cong&\bigoplus_{k=0}^n \left(T^{1|2}_{k+d-n}\times L^{2d|2n}_{\omega_{d-\frac{1}{2}(1+(-1)^k)}+\nu_k-\frac{1}{2}\nu_n}\right)\\
&&\oplus\bigoplus_{k=n+1}^\infty \left(T^{1|2}_{k+d-n}\times L^{2d|2n}_{(k-n)\epsilon_1+\omega_{d-\frac{1}{2}(1+(-1)^k)}+\frac{1}{2}\nu_n}\right).
\end{eqnarray*}

\section{The $\osp$-modules of spherical monogenics}
\label{Mkasreps}

In Corollary \ref{monosp} most of the spaces of spherical monogenics on $\mR^{m|2n}$ were identified as irreducible infinite dimensional highest weight modules of $\osp$. In this section we show that in the remaining cases the corresponding representations are not irreducible. They are still indecomposable highest weight representations and we determine their decomposition series.  An indecomposable representation is a representation which is not the direct sum of two subrepresentations. First we need the following lemma.
\begin{lemma}
\label{multfree}
The space of spherical monogenics $\cM_k$ on $\mR^{m|2n}$ has a multiplicity-free decomposition into simple $\mathfrak{so}(m)\oplus\mathfrak{sp}(2n)$-modules.
\end{lemma}
\begin{proof}
This can be proven similarly to the corresponding result for the spherical harmonics on superspace, see \cite{DBE1}. First the decomposition of $\cP\otimes\mS_{m|2n}$ under the action of $\mathfrak{so}(m)\oplus\mathfrak{sp}(2n)$ can be considered. This corresponds to the Fischer decompositions of $\mR[x_1,\cdots,x_m]\otimes\mS_{m|0}$ and $\Lambda_{2n}\otimes\mS_{0|2n}$, which can be found in \cite{MR2677004, MR1169463} and \cite{Krysl}. Each element of $\cP\otimes\mS_{m|2n}$ can be written in terms of the vector variables on $\mR^m$ and $\mR^{0|2n}$ and the corresponding spherical monogenics. Then it needs to be proven that for each pair of degrees of the spherical monogenics on $\mR^m$ and $\mR^{0|2n}$ there is only one polynomial in the two vector variables of a fixed degree, such that the product with the monogenics is super monogenic.

An alternative proof is to consider the space $\cH_k\otimes\mS^{+}_{2d|2n}$, using the decomposition into irreducible $\mathfrak{so}(2d)\oplus\mathfrak{sp}(2n)$-modules in Theorem 4 in \cite{DBE1}, and the decomposition $\mS^+_{2d|2n}=\mS^+_{2d|0}\times\mS^+_{0|2n}\oplus\mS^-_{2d|0}\times\mS^-_{0|2n}$, see \cite{Tensor}. Using the well-known classical tensor products for $\mathfrak{so}(2d)$ and $\mathfrak{sp}(2n)$ it follows that the decomposition of $\cH_k\otimes\mS^{+}_{2d|2n}$ has multiplicities not greater than two. It then remains to be checked that the representations that appear twice are split up under the decomposition $\cH_k\otimes\mS^{+}_{2d|2n}=\cM_k^+\oplus\bold{x}\cM_{k-1}^-$
\end{proof}

\begin{theorem}
\label{remMk}
For $\cM_k$ the space of spherical monogenics of homogeneous degree $k$ on $\mR^{2d|2n}$ with $d\le n$ and $1+n-d\le k\le 1+2n-2d$, the relations
\begin{eqnarray*}
\bold{x}^{2d-2n+2k-1}\cM_{2n-2d-k+1}&\subset&\cM_k
\end{eqnarray*}
and 
\begin{eqnarray*}
\cM^+_k/(\bold{x}^{2d-2n+2k-1}\cM^-_{2n-2d-k+1})&\cong&K^{2d|2n}_{k\epsilon_1+\omega_d-\frac{1}{2}\nu_n}\\
\cM^-_k/(\bold{x}^{2d-2n+2k-1}\cM^+_{2n-2d-k+1})&\cong&K^{2d|2n}_{k\epsilon_1+\omega_d+\nu_{n-1}-\frac{3}{2}\nu_n}
\end{eqnarray*}
hold. Furthermore $\cM_k^\pm$ is always indecomposable and
\begin{equation}
\label{onlysubrepMk}
\cM_k^\pm\cap\left(\bold{x}\,\,\cP\otimes\mS_{2d|2n}\right)=\bold{x}^{2d-2n+2k-1}\cM_{2n-2d-k+1}^\mp
\end{equation}
holds.
\end{theorem}
\begin{proof}
The first relation follows immediately from equation \eqref{px2l1}. 

For the second relation we take $1+n-d\le k\le 1+2n-2d$ and combine Lemma \ref{irrHk} with Theorem \ref{decomp2} into
\[
\left(\cH_{k+1}\otimes \mS^{-}_{2d|2n}\right)/\left(R^{2k+2d-2n}\cH_{1+2n-2d-k}\otimes \mS^{-}_{2d|2n}\right) \cong K^{2d|2n}_{(k+1)\epsilon_1+\omega_d+\nu_{n-1}-\frac{3}{2}\nu_n}\oplus K^{2d|2n}_{k\epsilon_1+\omega_d-\frac{1}{2}\nu_n}.
\]
Applying Theorem \ref{HkS} twice (which is possible since $k+1\not=1+n-d$ and $1+2n-2d-k\not=1+n-d$) then yields that this is isomorphic to
\begin{eqnarray*}
\left(\cM_{k+1}^-\oplus \bold{x}\cM_{k}^+ \right)/\left(\bold{x}^{2k+2d-2n}\cM_{1+2n-2d-k}^-\oplus\bold{x}^{2k+2d-2n+1}\cM^+_{2n-2d-k}\right).
\end{eqnarray*}
The first relation in the theorem implies that this is isomorphic to
\begin{eqnarray*}
\cM_{k+1}^-/\left(\bold{x}^{2k+2d-2n+1}\cM^+_{2n-2d-k}\right)\,\oplus\, \left(\bold{x}\cM_{k}^+\right)/\left(\bold{x}^{2k+2d-2n}\cM_{1+2n-2d-k}^-\right),
\end{eqnarray*}
where for $k=1+2n-2d$ the first quotient is equal to $\cM_{k+1}^-$ since we consider $\cM^+_{-1}=0$. This then proves the second part of the theorem by iteration.

Now we prove that $\cM_k^\pm$ is indecomposable. For $k$ satisfying $1+n-d\le k\le 1+2n-2d$, Lemma \ref{irrHk} implies that $\cH_{k+1}$ has a non-trivial submodule $R^{2k+2d-2n}\cH_{2n-2d+1-k}$. As an $\mathfrak{so}(2d)\oplus\mathfrak{sp}(2n)$-representation the finite dimensional $\cH_k$ is completely reducible, hence $R^{2k+2d-2n}\cH_{2n-2d+1-k}$ has a complement $\mathfrak{so}(2d)\oplus\mathfrak{sp}(2n)$-module:
 \[
\cH_{k+1}=U \oplus R^{2k+2d-2n}\cH_{2n-2d+1-k}\qquad\mbox{as $\mathfrak{so}(2d)\oplus\mathfrak{sp}(2n)$-representations}.
\]
Also as $\mathfrak{so}(2d)\oplus\mathfrak{sp}(2n)$-representations, the equality
\[
U\otimes \mS^+_{2d|2n} \cong K^{2d|2n}_{(k+1)\epsilon_1+\omega_d-\frac{1}{2}\nu_n}\oplus K^{2d|2n}_{k\epsilon_1+\omega_d+\nu_{n-1}-\frac{3}{2}\nu_n}
\]
holds, since $U\cong K^{2d|2n}_{(k+1)\epsilon_1}$ as an $\mathfrak{so}(2d)\oplus\mathfrak{sp}(2n)$-representation and $k+1\ge 1+n-d$, which implies that Theorem \ref{decomp2} can be applied. This defines two $\mathfrak{so}(2d)\oplus\mathfrak{sp}(2n)$-subrepresentations $V\subset \bold{x}\cM_k^-$ and $W\subset \cM_{k+1}^+$ such that, as an $\mathfrak{so}(2d)\oplus\mathfrak{sp}(2n)$-representations $V\cong K^{2d|2n}_{k\epsilon_1+\omega_d+\nu_{n-1}-\frac{3}{2}\nu_n}$ holds. Since $\cM_k^-$ has a multiplicity-free decomposition into $\mathfrak{so}(2d)\oplus\mathfrak{sp}(2n)$-representations, this implies that $V$ is the unique $\mathfrak{so}(2d)\oplus\mathfrak{sp}(2n)$-complement module of $\bold{x}^{2d-2n+2k}\cM^+_{2n-2d-k+1}$ inside $\bold{x}\cM_k^-$. It remains to be proved that the decomposition \[\bold{x}\cM_k^-=V\oplus \bold{x}^{2d-2n+2k}\cM^+_{2n-2d-k+1}\] does not hold as $\mathfrak{osp}(2d|2n)$-modules. This is equivalent to proving that the decomposition
\begin{eqnarray*}
\cH_{k+1}\otimes\mS^+_{2d|2n}=\bold{x}^{2d-2n+2k}\cM^+_{2n-2d-k+1}\oplus Z&\mbox{with}&Z=V\oplus\cM_{k+1}^+
\end{eqnarray*} does not hold as $\mathfrak{osp}(2d|2n)$-modules.

The highest weight vector of $\bold{x}^{2d-2n+2k}\cM^+_{2n-2d-k+1}$ is $v_2^+\otimes 1$ with $v_2^+$ the highest weight vector of $R^{2n-2n+2k}\cH_{1+2n-2d-k}$. This highest weight vector (vector which is annihilated by all positive root vectors) $v_2^+\in R^{2k+2d-2n}\cH_{2n-2d+1-k}\subset\cH_{k+1}$ is generated by action of negative root vectors on other elements of $\cH_k$, since $\cH_k$ is an indecomposable highest weight module. This highest weight vector can therefore be expressed as $v^+_2=\sum_{i}Y_i u_i$ with $u_i\in U$ and $Y_i$ negative root vectors of $\mathfrak{osp}(2d|2n)$. This means that we can write the highest weight vector of $\bold{x}^{2d-2n+2k}\cM^+_{2n-2d-k+1}$ as
\[
v_2^+\otimes 1 =\sum_i Y_i(u_i\otimes 1)-\sum_i (-1)^{|Y_i||u_i|}u_i\otimes Y_i(1).
\]
Now, $\sum_i (-1)^{|Y_i||u_i|}u_i\otimes Y_i(1)\in U\otimes\mS^+_{2d|2n}\subset Z$ and $(u_i\otimes 1)\in Z$ as well. If $Z$ is an $\mathfrak{osp}(2d|2n)$-module, then $\sum_i Y_i(u_i\otimes 1)\in Z$ also holds, which would imply that $v_2^+\otimes 1 \in Z$ which is a contradiction. Therefore $Z$ is not an $\mathfrak{osp}(2d|2n)$-module. The proof for $\cM_k^+$ is exactly the same.

Since $2n-2d-k+1\le n-d$, Corollary \ref{monosp} implies that $\bold{x}^{2d-2n+2k-1}\cM_{2n-2d-k+1}^-$ is an irreducible subrepresentation. The fact that $\cM_k^+/\bold{x}^{2d-2n+2k-1}\cM_{2n-2d-k+1}^-$ is irreducible implies that $\cM_k^+$ has no other submodule. Therefore $\cM_k^\pm$ has only one submodule and since $\cM_k^\pm\cap\left(\bold{x}\cP\otimes\mS_{2d|2n}\right)$ is also an $\mathfrak{osp}(2d|2n)$-module it must be equal to $\bold{x}^{2d-2n+2k-1}\cM_{2n-2d-k+1}^-$.
\end{proof}

This theorem yields, as a side result, the proof that Theorem 9 in \cite{Tensor} constitutes the complete decomposition series of the tensor product $K^{2d|2n}_{(n-d+1)\epsilon_1}\otimes K^{2d|2n}_{\omega_d-\frac{1}{2}\nu_n}$. This could not be settled in \cite{Tensor} and is stated in the following corollary, which gives extra information on the exceptional case in Theorem \ref{decomp1} and \ref{decomp2}.
\begin{corollary}
\label{tensornotcr}
If $n\ge d$ the tensor products $K^{2d|2n}_{(n-d+1)\epsilon_1}\otimes \mS^\pm_{2d|2n}$ are indecomposable but not irreducible. The representation has subrepresentations
\begin{eqnarray*}
K^{2d|2n}_{(n-d+1)\epsilon_1}\otimes K^{2d|2n}_{\omega_d-\frac{1}{2}\nu_n}\supsetneq V\supsetneq K^{2d|2n}_{(n-d)\epsilon_1+\omega_d+\nu_{n-1}-\frac{3}{2}\nu_n},
\end{eqnarray*}
with $V$ an indecomposable representation satisfying
\begin{eqnarray*}
\left(K^{2d|2n}_{(n-d+1)\epsilon_1}\otimes K^{2d|2n}_{\omega_d-\frac{1}{2}\nu_n}\right)/V&\cong&K^{2d|2n}_{(n-d)\epsilon_1+\omega_d+\nu_{n-1}-\frac{3}{2}\nu_n}\\
V/ K^{2d|2n}_{(n-d)\epsilon_1+\omega_d+\nu_{n-1}-\frac{3}{2}\nu_n}&\cong & K^{2d|2n}_{(n-d+1)\epsilon_1+\omega_d-\frac{1}{2}\nu_n}
\end{eqnarray*}
and the statements for $\mS^-_{2d|2n}$ are similar.
\end{corollary}
\begin{proof}
The identifications $\mS^+_{2d|2n}\cong K^{2d|2n}_{\omega_d-\frac{1}{2}\nu_n}$ and $\cH_{n-d+1}\cong K^{2d|2n}_{(n-d+1)\epsilon_1}$ hold. Corollary \ref{monosp} shows that $\bold{x}\cM_{n-d}^-\cong K^{2d|2n}_{(n-d)\epsilon_1+\omega_d+\nu_{n-1}-\frac{3}{2}\nu_n}$ holds. Then we define $V=\cM_{n-d+1}^+$. Theorem \ref{remMk} shows that 
\[V/ K^{2d|2n}_{(n-d)\epsilon_1+\omega_d+\nu_{n-1}-\frac{3}{2}\nu_n}\cong  K^{2d|2n}_{(n-d+1)\epsilon_1+\omega_d-\frac{1}{2}\nu_n}\] holds. The corollary is therefore proved if 
\begin{eqnarray*}
\left(\cH_{n-d+1}\otimes \mS^+_{2d|2n}\right)/\cM_{n-d+1}^+&\cong& \cM_{n-d}^-
\end{eqnarray*}
holds. This identity follows immediately from considering the operator $\px$ on $\cH_{n-d+1}\otimes \mS^+_{2d|2n}$, since Im$(\px)=\cM_{n-d}^-$ while Ker$(\px)=\cM_{n-d+1}^+$.
\end{proof}

%\begin{remark}
%\label{casimir}
%\rm{The theory in this paper gives even more results on the tensor product 
%\[K^{2d|2n}_{(n-d+1)\epsilon_1}\otimes K^{2d|2n}_{\omega_d-\frac{1}{2}\nu_n}\cong\cH_{n-d+1}\otimes\mS^+_{2d|2n}.\]
%It can be explicitly checked that the operator $\bold{x}\px \,:\,\cH_{n-d+1}\otimes\mS^+_{2d|2n}\to \cH_{n-d+1}\otimes\mS^+_{2d|2n}$ is (up to an additive constant) the quadratic Casimir operator of $\osp$ for this $\osp$-realization.

%This shows that the Casimir operator is not diagonalizable on this representation. The only eigenvalue of $\bold{x}\px$ is zero since
%\begin{eqnarray*}
%(\bold{x}\px)^2  =-R^2\Delta -\bold{x}(2\mE+2d-2n)\px=0
%\end{eqnarray*}
%on  $\cH_{n-d+1}\otimes\mS^+_{2d|2n}$. Not all elements of $\cH_{n-d+1}\otimes\mS^+_{2d|2n}$ have eigenvalue zero, since Im$(\bold{x}\px)$ on $\cH_{n-d+1}\otimes\mS^+_{2d|2n}$ is $\bold{x}\cM^-_{k-1}$.

%For the classification of the conformally invariant operators, see Section \ref{Fegan}, the spaces Im($\cC_2+\mu)$ for $\mu\in\mR$ are of importance. The only possible spaces of the form Im($\cC_2+\mu)$ for $\cC_2$ the quadratic Casimir operator of $\osp$ and $\mu$ a real constant are $K^{2d|2n}_{(n-d)\epsilon_1+\omega_d+\nu_{n-1}-\frac{3}{2}\nu_n}$ and the full representation $K^{2d|2n}_{(n-d+1)\epsilon_1}\otimes K^{2d|2n}_{\omega_d-\frac{1}{2}\nu_n}$.}
%\end{remark}
 
\begin{remark}
\label{casimir}
It can be checked that up to an additive constant, the quadratic Casimir operator on $\cH_k \otimes \mS^{m|2n}$ is given by the operator $\bold{x}\px$. In particular this shows that whenever $K^{m|2n}_{k\epsilon_1} \otimes \mS^{m|2n}$ is completely reducible the Casimir operator has two different eigenvalues which can easily be checked directly. But it also shows that in case $K^{m|2n}_{k\epsilon_1} \otimes \mS^{m|2n}$ is not completely reducible, the Casimir operator is not diagonalizable.
\end{remark} 

\section{Symmetries of the super Dirac operator}
\label{secCon}

In this section we construct all first order generalized symmetries of the super Dirac operator with scalar symbol. Generalized symmetries of the super Dirac operator are differential operators $D$ for which there exists another differential operator $\delta$ such that
\begin{eqnarray*}
\px D=\delta \px
\end{eqnarray*}
holds. Such operators clearly preserve the kernel of the Dirac operator. The symmetries which are first order generate a Lie superalgebra since the composition of two symmetries is still a symmetry and the super commutator of two first order differential operators is still first order.

We define first order differential operators on $\cO(\mR^{m|2n})\otimes \mS_{m|2n}$ to be elements of the vector space
\begin{eqnarray*}
\left( \cO(\mR^{m|2n})\otimes\sC)\right)\,\oplus\,\left(\mathfrak{vect}(m|2n)\otimes\sC\right)
\end{eqnarray*}
with $\mathfrak{vect}(m|2n)$ the Lie superalgebra of vectorfields on $\cO(\mR^{m|2n})$, endomorphisms which satisfy the graded Leibniz rule.

We will find that the first order generalized symmetries with scalar symbol generate the Lie superalgebra $\mathfrak{osp}(m+1,1|2n)$, which is defined as in Section \ref{prelosp} but with a metric such that the orthogonal part has signature $m+1,1$. Therefore the kernel of the Dirac operator has the structure of an $\mathfrak{osp}(m+1,1|2n)$-module. We prove that this module is irreducible if $m-2n\not\in2-2\mN$. If $m-2n\in2-2\mN$ it is reducible but indecomposable and we determine the decomposition series.

In the classical case of the Dirac operator $\upx$ on $\mR^m$, the algebra of first order generalized symmetries with scalar symbol is isomorphic to the Lie algebra $\mathfrak{so}(m+1,1)$, which is also the algebra of conformal Killing vector fields on $\mR^m$. More precisely, the leading term of the generalized symmetries is the corresponding conformal Killing vector field. This conformal invariance of $\upx$ follows immediately from the similarity between the Stein-Weiss construction in \cite{MR0223492} and the construction of conformally invariant first order differential operators by Fegan in \cite{MR0482879}. By extending the calculation of Killing vector fields on $\mR^{m|2n}$ in Section 3 of \cite{OSpHarm} it can be proved 
 that the conformal algebra for $\mR^{m|2n}$ is $\mathfrak{osp}(m+1,1|2n)$. Again the leading terms of the generalized symmetries correspond to conformal Killing vector fields. Therefore we say that the super Dirac operator on $\mR^{m|2n}$ is conformally invariant.

To describe the conformal symmetries we introduce the following Kelvin inversion 
\[I:\cO(\mR^{m|2n}_0)\to\cO(\mR^{m|2n}_0),
\] 
with $\cO(\mR^{m|2n}_0)=\cC^\infty(\mR^m_0)\otimes\Lambda_{2n}$, via
\[
(If)(\bold{x})=\bold{x}R^{-M}f\left(\frac{\bold{x}}{R^2}\right),
\]
where $f\left(\frac{\bold{x}}{R^2}\right)$ should be understood as a finite Taylor expansion in the anticommuting variables, see e.g. \cite{CDBS3} for a very explicit approach to such functions. This Kelvin inversion satisfies $I^2=-1$ and therefore is an isomorphism of $\cO(\mR^{m|2n}_0)$.
\begin{theorem}
\label{confSymm}
The super Dirac operator on $\mR^{m|2n}$ has a Lie superalgebra of generalized symmetries which is isomorphic to $\mathfrak{osp}(m+1,1|2n)$. This realization of $\mathfrak{osp}(m+1,1|2n)$ is given by the differential operators
\begin{eqnarray}
\label{Pij}
\Pi_j=\bold{x}\widehat{E}_j+X_j(M+2\mE)-R^2\partial_{X^j}&\mbox{and}&\quad \partial_{X^j} \qquad\mbox{ for}\quad j=1,\cdots,m+2n,\\
\nonumber
2\mE+M-1 &\mbox{and}&\quad K_{ij}\quad\mbox{ given in equation \eqref{ospDirac}.}
\end{eqnarray}
\end{theorem}
\begin{proof}
As mentioned before the differential operators $K_{ij}$ in equation \eqref{ospDirac} commute with the Dirac operator. The partial derivatives $\partial_{X^j}$ clearly also commute with the super Dirac operator $\px$. Next we prove that the differential operators $\Pi_j$ in equation \eqref{Pij} are generalized symmetries. Therefore we calculate, using Theorem \ref{osp12} and Proposition \ref{DiracCliff}
\begin{eqnarray*}
\px\Pi_j&=&\px\bold{x}\widehat{E}_j+\px X_j(M+2\mE)-\px R^2\partial_{X^j}\\
&=&-\bold{x}\px\widehat{E}_j-(2\mE+M)\widehat{E}_j+\widehat{E_j}(M+2\mE)+X_j(M+2\mE+2)\px-2\bold{x}\partial_{X^j}-R^2\partial_{X^j}\px\\
&=&2\bold{x}\partial_{X^j}+\bold{x}\widehat{E_j}\px+X_j(M+2\mE+2)\px-2\bold{x}\partial_{X^j}-R^2\partial_{X^j}\px\\
&=&\left(\bold{x}\widehat{E_j}+X_j(M+2\mE+2)-R^2\partial_{X^j}\right)\px,
\end{eqnarray*}
which implies that $\Pi_j$ is a generalized symmetry of $\px$. A direct calculation shows that these generalized symmetries can be written as
\begin{eqnarray}
\label{PiI}
\Pi_j=I\circ \partial_{X^j}\circ I,
\end{eqnarray}
where $I$ is the Kelvin inversion.

The differential operator $2\mE+M-1$ is also clearly a generalized symmetry.

Now we define operators $K_{\alpha\beta}$ for $\alpha,\beta=-1,0,1,\cdots,m+2n$ with $\alpha\le\beta$ given by:
\[K_{\alpha\beta}=\begin{cases}
\mE+\frac{M-1}{2} &\mbox{for } \alpha=-1 \mbox{ and }\beta=0\\
\frac{1}{2}\left(\Pi_j-\partial_{X^j}\right)&\mbox{for }\alpha=-1 \mbox{ and }\beta=j>0\\
\frac{1}{2}\left(\Pi_j+\partial_{X^j}\right)&\mbox{for }\alpha=0\mbox{ and }\beta=j>0
\end{cases}
\]
and equal to the operators $K_{ij}$ in equation \eqref{ospDirac} for $\alpha,\beta>0$ and $\alpha=i$ and $\beta=j$.

We introduce the metric $h\in\mR^{(m+2+2n)\times(m+2+2n)}$ given by
\[
h=\left( \begin{array}{ccc} -1&0&0\\ 0&1&0\\  0&0&g\end{array}\right)
\]
with $g$ the metric in equation \eqref{gmetric}.

Now we prove that the operators $K_{\alpha\beta}$ satisfy 
\[
[K_{\alpha\beta},K_{\gamma\delta}]=h_{\gamma\beta}K_{\alpha\delta}+(-1)^{[\alpha]([\beta]+[\gamma])}h_{\delta\alpha}K_{\beta\gamma}-(-1)^{[\gamma][\delta]}h_{\delta\beta}K_{\alpha\gamma}-(-1)^{[\alpha][\beta]}h_{\gamma\alpha}K_{\beta\delta}
\]
with $[-1]=[0]=\overline{0}$ and $[i]$ given by equation \eqref{gradmap} for $i>0$. It is easy to verify $[K_{-1,0},K_{-1,j}]=K_{0,j}$, $[K_{-1,0},K_{0,j}]=K_{-1,j}$ and $[K_{-1,0},K_{ij}]=0$. Then we calculate
\begin{eqnarray*}
[\partial_{X^k},\Pi_j]&=&\widehat{E_k}\widehat{E_j}+g_{jk}(M+2\mE)+2(-1)^{[j][k]}X_j\partial_{X^k}-2X_k\partial_{X^j}\\
&=&-2\left(K_{kj}+g_{jk}(\mE+\frac{M-1}{2})\right).
\end{eqnarray*}
Together with $[\partial_{X^j},\partial_{X^k}]=0$ and $[\Pi_j,\Pi_k]=0$, which is a consequence of equation \eqref{PiI}, this leads to the relations $[K_{-1j},K_{-1k}]=K_{jk}$, $[K_{0j},K_{0k}]=-K_{jk}$ and $[K_{-1j},K_{0k}]=-g_{kj}K_{-10}$. To obtain the remaining commutation relations we need
\begin{eqnarray*}
[\partial_{X^j},K_{kl}]=[\partial_{X^j},L_{kl}]=g_{kj}\partial_{X^l}-(-1)^{[k][l]}g_{lj}\partial_{X^k},
\end{eqnarray*}
which follows from equation \eqref{partL}, and
\begin{eqnarray*}
[\Pi_j,K_{kl}]&=&[\bold{x}\widehat{E}_j,K_{kl}]+[X_j(M+2\mE),L_{kl}]-[R^2\partial_{X^j},L_{kl}]\\
&=&\bold{x}[\widehat{E}_j,B_{kl}]+[X_j,L_{kl}](M+2\mE)-R^2[\partial_{X^j},L_{kl}]\\
&=&\bold{x}\left(g_{kj}\widehat{E_l}-(-1)^{[j][k]}g_{lj}\widehat{E_k}\right)+\left(g_{kj}X_l-(-1)^{[j][k]}g_{lj}X_k\right)(M+2\mE)\\
&-&R^2\left(g_{kj}\partial_{X^l}-(-1)^{[k][l]}g_{lj}\partial_{X^k}\right)\\
&=&g_{kj}\Pi_l-(-1)^{[k][l]}g_{lj}\Pi_{k},
\end{eqnarray*}
where we used equation \eqref{Baction} and the $\osp$-invariance of $\mE$, $\bold{x}$ and $R^2$.
\end{proof}

In \cite{Joseph}, the action of the complexified $\mathfrak{osp}(m+1,1|2n)$-algebra on scalar functions on $\mR^{m|2n}$ is studied. We note that the action of $\mathfrak{osp}(m+1,1|2n)$ on $\cO(\mR^{m|2n})\otimes\mS_{m|2n}$ considered in the current paper does not correspond to a tensor product action of the action of $\mathfrak{osp}(m+1,1|2n)$ on $\cO(\mR^{m|2n})$ with some action on $\mS_{m|2n}$, contrary to the $\osp$-action. 

Now we prove that the symmetries in Theorem \ref{confSymm} constitute all first order symmetries with scalar symbol.
\begin{theorem}
Every first order generalized symmetry with scalar highest order term of the super Dirac operator $\px$ on $\mR^{m|2n}$ is included in the realization of $\mathfrak{osp}(m+1,1|2n)$ in Theorem \ref{confSymm}.
\end{theorem}
\begin{proof}
We need to classify all differential operators
\[D=\sum_{j=1}^{m+2n}F_j\partial_{X_j}+F_0\]
with $F_j \in \cO(\mR^{m|2n})$ (not all zero) and $F_0\in\cO(\mR^{m|2n})\otimes \sC$ such that $\px D=\delta \px$ for $\delta$ another differential operator. Since $\px$ is a homogeneous operator, we can assume that $D$ is homogeneous and therefore use the notation $|D|=|F_0|=|F_j|+[j]$.

First we restrict to the case $F_0\in\cO(\mR^{m|2n})$. The condition for $D$ to be a generalized symmetry then becomes $\px(F_0)=0$ (which has to be regarded as an equation in $\cO(\mR^{m|2n})\otimes\sC$, not in $\cO(\mR^{m|2n})\otimes\mS_{m|2n}$) and $\sum_{k=1}^{m+2n}\widehat{E}_k(\partial_{X_k}F_j)=H\widehat{E}_j$ for some $H\in \cO(\mR^{m|2n})\otimes \sC$. The first condition implies that $F_0$ is a constant. Since the left-hand term of the second equation is vector valued, $H\widehat{E}_j$ has to be vector valued, which implies that $H$ is scalar. If follows that $(\partial_{X_k}F_j)=(-1)^{|D|[j]}\delta_j^kH$. Acting with a second derivative $\partial_{X_i}$ on this identity shows that $H$ is constant and therefore that $F_j$ is an element of $\cP_0\oplus\cP_1$. This yields the differential operators, (up to additive constants) $D=\partial_{X_j}$ and $D=\mE$, which are included in Theorem \ref{confSymm}.

Now we consider $F_0\in \left(\cO(\mR^{m|2n})\otimes \sC\right)\,\backslash \,\cO(\mR^{m|2n})$. This leads to $\px(F_0)=0$ and 
\[\sum_{k}\widehat{E}_k(\partial_{X_k}F_j)=H\widehat{E}_j-(-1)^{[j]|D|}\widehat{E}_jF_0\] for some $H\in \cO(\mR^{m|2n})\otimes \sC$. Only the vector valued term in $H\widehat{E}_j-(-1)^{[j]|D|}\widehat{E}_jF_0$ can be different from zero and influence $F_j$, the others therefore have to cancel out between $F_0$ and $H$ and lead to independent zero-order differential operators which we do not take into account. The non-trivial contributions of $F_0$ can only come from scalars and bi-vectors, which can be seen from equation \eqref{commrelCliff}. Therefore we can expand $F_0$ in a unique way as $F_0=\sum_{k,l}f_{kl}B_{kl}+f_0$, with $f_0\in\cO(\mR^{m|2n})$ and $f_{kl}=-(-1)^{[k][l]}f_{lk}\in\cO(\mR^{m|2n})$. We obtain, using equation \eqref{Baction}, that
\begin{eqnarray*}
H\widehat{E}_j-(-1)^{[j]|D|}\widehat{E}_jF_0&=&H\widehat{E}_j-F_0\widehat{E}_j+\sum_{k,l}f_{kl}[B_{kl},\widehat{E}_j]\\
&=&h\widehat{E}_j+2\sum_{k,l}f_{kl}g_{jl}\widehat{E}_k
\end{eqnarray*}
holds, where $h=H-F_0$ again has to be scalar. The conditions on $D=\sum_jF_j\partial_{X_j}+\sum_{k,l}f_{kl}B_{kl}+f_0$ then become
\begin{eqnarray}
\label{setofeq1}
\px(\sum_{k,l}f_{kl}B_{kl}+f_0)&=&0\\
\label{setofeq2}
\partial_{X_k}F_j&=&h\delta_{jk}+2\sum_{l}f_{kl}g_{jl}
\end{eqnarray}
for an arbitrary $h\in\cO(\mR^{m|2n})$.

First we prove that equation \eqref{setofeq1} leads to the restriction that $F_0=\sum_{k,l}f_{kl}B_{kl}+f_0$ is an element of $(\cP_0\oplus\cP_1)\otimes\sC$. The equation $\px(F_0)=0$ falls apart into a vector and a tri-vector part,
\begin{eqnarray}
\nonumber
\partial_{X_k}f_0&=&\sum_{j,l}g_{jl}\partial_{X_j}f_{kl} (-1)^{([k]+[l])(|D|+1)]}\\
\label{pffff}
(-1)^{[j](|D|+[j])}\partial_{X_j} f_{kl}&=&-(-1)^{[k](|D|+[k])}\partial_{X_k} f_{lj}(-1)^{[j]([k]+[l])}.
\end{eqnarray}
>From the second equation it follows that $\partial_{X_i}\partial_{X_j}f_{kl}=0$, so $f_{kl}$ is first order and then the first equation shows that $f_0$ is first order as well. 

As a next step we prove that $F_j\in\cP_2\oplus\cP_1\oplus\cP_0$. The condition $\partial_{X_k}F_j=h\delta_{jk}+2\sum_{l}f_{kl}g_{jl}$ with $f_{kl}$ of degree 1 implies $\partial_{X_i}\partial_{X_l}\partial_{X_k}F_j=\left(\partial_{X_i}\partial_{X_l}h\right)\delta_{jk}$ which yields that $h$ must be of degree 1 and $F_j$ of degree 2. Equations \eqref{setofeq1} and \eqref{setofeq2} show that there is no mixing up of different degrees, i.e. we can consider $F_j\in\cP_2$, $F_j\in\cP_1$ and $F_j\in\cP_0$ independently. First we take $F_j\in\cP_0$. Equation \eqref{setofeq2} then implies that $f_{kl}=-hg_{kl}$ and since $f_{kl}=-(-1)^{[k][l]}f_{lk}$ this implies $f_{kl}=0$ and we arrive in the case $F_0$ scalar which is already dealt with. Now assume $F_j\in\cP_1$. Then equation \eqref{setofeq2} together with the anti-symmetry of $f_{kl}$ imply that $f_{kl}$ are constants and therefore $F_0$ is constant.
  It can then easily be checked that exactly the symmetries $K_{ij}=L_{ij}+B_{ij}$ are obtained. 

Finally take $F_j\in \cP_2$. Then $h\in\cP_1$ and up to a choice of coordinates and a renormalization we can assume $h=2X_l$ and we study the generalized symmetry $D-\Pi_l=\sum_j \widetilde{F}_j\partial_{X_j}+\widetilde{F}_0$ with $\Pi_l$ defined in Theorem \ref{confSymm}. For this generalized symmetry, equation \eqref{setofeq2} becomes
\[
\partial_{X_k}\widetilde{F}_j=2\sum_{l}\widetilde{f}_{kl}g_{jl}.
\]
The combination of this equation together with equation \eqref{pffff} for $\widetilde{f}_{kl}$ then shows that $\partial_{X_l}\partial_{X_k}\widetilde{F}_j=0$ which shows that $D=\Pi_l$ up to zero degree terms.
\end{proof}

Since the differential operators in $\mathfrak{osp}(m+1,1|2n)$ are generalized symmetries of the Dirac operator, the kernel constitutes a module. In the following theorem we study this kernel $\cM=\bigoplus_{k=0}^\infty\cM_k$.

\begin{theorem}
\label{bigmonrep}
The space of monogenic polynomials $\cM=\bigoplus_{k=0}^\infty\cM_k$ on $\mR^{m|2n}$ is a irreducible\\ $\mathfrak{osp}(m+1,1|2n)$-module if $m$ is odd with $\mathfrak{osp}(m+1,1|2n)$-action given in Theorem \ref{confSymm}. If $M$ is even and strictly positive, the spaces $\cM^+=\bigoplus_{k=0}^\infty\cM^+_k$ and $\cM^-=\bigoplus_{k=0}^\infty\cM^-_k$ are irreducible $\mathfrak{osp}(m+1,1|2n)$-modules. If $M=-2p$, with $p\in\mN$, then the spaces $\cM^+$ and $\cM^-$ are still indecomposable module, but they have a submodule, given by
\[
\bigoplus_{k=0}^{p}\cM^\pm_k\oplus\bigoplus_{k=p+1}^{2p+1}\bold{x}^{2k-2p-1}\cM^\mp_{2p-k+1}\cong K^{2d+2|2d+2p}_{p\epsilon_1+\omega_{d+1}-\frac{1}{2}\nu_{d+p}-\frac{1}{2}(1\mp 1)\delta_{d+p} }.
\]
\end{theorem}
\begin{proof}
First we consider the case $m$ odd and look at the action of $\osp\hookrightarrow\mathfrak{osp}(m+1,1|2n)$ given by the operators $K_{ij}$ in \eqref{ospDirac}. The $\mathfrak{osp}(m+1,1|2n)$-representation $\cM$ decomposes into irreducible $\osp$-representations as $\cM=\bigoplus_{k=0}^\infty\cM_k$, see Corollary \ref{monosp}. Because of the partial derivatives $\partial_{X^j}\in\mathfrak{osp}(m+1,1|2n)$ it is clear that this representation is indecomposable and that each subrepresentation will contain the scalars $\cM_0$. Since we know that $\Pi_1^k1\in\cM_k$ the theorem is proved if we can show that $\Pi_1^k1\not=0$ for all $k\in\mN$. Because of equation \eqref{PiI} this is equivalent to proving $\partial^k_{X^1}\bold{x}R^{-M}\not=0$, which follows immediately.

For the case $m$ even, the structure of the differential operators in $\mathfrak{osp}(m+1,1|2n)$ shows that $\cM$ decomposes as $\cM^+\oplus\cM^-$ as representations. We focus on $\cM^+$, the other case being completely similar.  For any value of $M$ this representation decomposes as
\begin{eqnarray*}
\cM^+=\bigoplus_{k=0}^\infty\cM_k^+
\end{eqnarray*}
into indecomposable $\osp$-modules, see Corollary \ref{monosp} and Theorem \ref{remMk}. The partial derivatives again imply that $\cM$ is an indecomposable $\mathfrak{osp}(m+1,1|2n)$-module and that each subrepresentation contains the scalars. If $M>0$, the proof is the same as in the case $M$ odd. 

Now we focus on the case $M=-2p$. Then we easily find that $\Pi_1^k1\not=0$ if $k\le 2p+1$ since $\bold{x}R^{-M}$ is a polynomial in $X_1$ of maximal degree $1+2p$. So for each subrepresentation $U\subset\cM^+$ we find $U\cap\cM^+_k\not=0$ if $k\le 2p+1$. By considering the smallest $\osp$-subrepresentations of these $\cM_k^+$ according to Corollary \ref{monosp} and Theorem \ref{remMk}, this implies that
\begin{eqnarray*}
\bigoplus_{k=0}^p\cM^+_k\oplus\bigoplus_{k=p+1}^{2p+1}\bold{x}^{2k-2p-1}\cM^-_{2p-k+1}&\subset&U
\end{eqnarray*}
holds. Next we prove that there can be no other $\osp$-modules included in $U$ than those on the left-hand in the equation above. Assume $\cM_k^+\subset U$ for $k>p$ and take an arbitrary $M_k\in\cM_k^+$ with $M_k\not\in\bold{x}\cP\otimes\mS_{m|2n}$. Then it is obvious that $\Pi_jM_k\not\in\bold{x}^2\cP\otimes\mS_{m|2n}$. Because of equation \eqref{onlysubrepMk} this implies that $\Pi_jM_k\not\in\bold{x}\cP\otimes\mS_{m|2n}$. By induction it follows that all $\cM^+_k$ are inside $U$ so $U=\cM$. Therefore the only possible $U$ is of the form $\bigoplus_{k=0}^p\cM^+_k\oplus\bigoplus_{k=p+1}^{2p+1}\bold{x}^{2k-2p-1}\cM^-_{2p-k+1}$.

The last step is to prove that this is actually an $\mathfrak{osp}(m+1,1|2n)$-module. This corresponds to showing that the action of $\Pi_j$ and $\partial_{X^j}$ stabilizes $U$. The relations
\begin{eqnarray*}
\Pi_j\cM^+_{p}\subset \bold{x}\cM^-_p&\mbox{and}&\Pi_j\bold{x}^{2i-1}\cM^-_{p-i+1}\subset \bold{x}^{2i+1}\cM^-_{p-i}\quad\mbox{and}\quad \Pi_j\bold{x}^{2p+1}=0,
\end{eqnarray*}
are consequence of equations \eqref{Pij} and \eqref{onlysubrepMk} and the corresponding claims for $\partial_{X^j}$ follow similarly. The identification of the highest weight of this irreducible representation follows easily.
\end{proof}

As a side result of Theorem \ref{bigmonrep} and Corollary \ref{monosp} we obtain the following branching rule.
\begin{corollary}
The following branching rule of $\mathfrak{osp}(2d|2d+2p)\hookrightarrow\mathfrak{osp}(2d+2|2d+2p)$ holds:
\begin{eqnarray*}
K^{2d+2|2d+2p}_{p\epsilon_1+\omega_{d+1}-\frac{1}{2}\nu_{d+p}}&\cong&\bigoplus_{k=0}^{p}\left(K^{2d|2d+2p}_{k\epsilon_1+\omega_{d}-\frac{1}{2}\nu_{d+p}}\oplus K^{2d|2d+2p}_{k\epsilon_1+\omega_{d}+\nu_{d+p-1}-\frac{3}{2}\nu_{d+p}}\right).
\end{eqnarray*}
\end{corollary}

Two limit cases are
\begin{eqnarray*}
K^{2|2p}_{\left(p+\frac{1}{2}\right)\epsilon-\frac{1}{2}\nu_{p}}&\cong&\bigoplus_{k=0}^{p}\left(K^{0|2p}_{\nu_k-\frac{1}{2}\nu_{p}}\oplus K^{0|2p}_{\nu_k+\nu_{p-1}-\frac{3}{2}\nu_{p}}\right) \mbox{ for } \mathfrak{sp}(2p)\hookrightarrow\mathfrak{osp}(2|2p)\mbox{ and}\\
K^{2d+2|2d}_{\omega_{d+1}-\frac{1}{2}\nu_{d}}&\cong&K^{2d|2d}_{\omega_{d}-\frac{1}{2}\nu_{d}}\oplus K^{2d|2d}_{\omega_{d}+\nu_{d-1}-\frac{3}{2}\nu_{d}}\mbox{ for } \mathfrak{osp}(2d|2d)\hookrightarrow\mathfrak{osp}(2+2d|2d).
\end{eqnarray*}
A method to calculate the highest weights of the first case is explained in \cite{Tensor}. The second one is just a special case of the equality $\mS_{2d|2n}^+\oplus\theta_{d+1}\mS_{2d|2n}^-=\mS_{2d+2|2n}^+$.

In this section we investigated the first order generalized symmetries of the super Dirac operator with scalar symbol, which lead to the Lie superalgebra $\mathfrak{osp}(m+1,1|2n)$. This forms but the first step in the program of determining all (higher order) generalized symmetries of this operator. For the ordinary Laplace operator on $\mR^m$, this problem was solved in the seminal paper \cite{Eastwood}. In this case, it turns out that the higher symmetries constitute the algebra $\cU(\mathfrak{so}(m+1,1))/\cJ$ with $\cJ$ the Joseph ideal. Some preliminary results on the higher symmetries of the super Laplace operator are already obtained in \cite{Joseph}. It is expected that the algebra of symmetries will be $\cU(\mathfrak{osp}(m+1,1|2n))/\cJ$ with $\cJ$ an ideal in the universal enveloping algebra studied in \cite{Joseph}.  The generalized symmetries of the classical Dirac operator are studied in \cite{ESS}. It is an interesting open question whether these results can be 
 generalized to the super Dirac operator, introduced in the present paper.

%In this section we investigated the first order generalized symmetries of the super Dirac operator, which lead to the Lie superalgebra $\mathfrak{osp}(m+1,1|2n)$. In \cite{ESS} the higher order generalized symmetries of the classical Dirac operator are studied. It is an interesting open question whether these results can be generalized to the super Dirac operator in the current paper. In \cite{Eastwood} the higher symmetries of the Laplace operator are studied, there they constitute the algebra $\cU(\mathfrak{so}(m+1,1))/\cJ$ with $\cJ$ the Joseph ideal. Some results on the higher symmetries of the super Laplace operator are obtained in \cite{Joseph}, it is expected that the algebra of symmetries will be $\cU(\mathfrak{osp}(m+1,1|2n))/\cJ$ with $\cJ$ an ideal in the universal enveloping algebra studied in \cite{Joseph}.

\section{The Fegan classification}
\label{Fegan}

In \cite{MR0482879}, the classification of conformally invariant first order differential operators between functions on $\mR^m$ or $\mS^m$ with values in finite dimensional irreducible $\mathfrak{so}(m)$-modules, was obtained. The conformal symmetries are given by the Lie algebra $\mathfrak{so}(m+1,1)$ which has a $\mZ$-gradation $\mathfrak{so}(m+1,1)=\mR^m+\mathfrak{co}(m)+\mR^m$, with $\mathfrak{co}(m)=\mathfrak{so}(m)\oplus\mR$. There exists such an operator if and only if $L^m_\mu$ appears in the decomposition into irreducible pieces of the tensor product $\mC^m\otimes L^m_\lambda$. This operator $D$ is the composition of the gradient and the invariant projection $\mC^m\otimes L^m_\lambda\to L^m_\mu$, which is clearly $\mathfrak{so}(m)$-invariant. The tensor product is always multiplicity-free. For each representation $L^m_\mu$ appearing in the tensor product of $L^m_\lambda$, there is a unique conformal weight (character of the one dimensional Lie algebra in $\mathfrak{
 co}(m)$)
 , making them into $\mathfrak{co}(m)+\mR^m$-representations, such that the operator $D$ is conformally invariant. In particular, if the spinor spaces $\mS_m^{(\pm)}$ are considered, the tensor product contains two irreducible representations, see Theorem \ref{decomp1}. The two invariant first order differential operators then are the Dirac operator, corresponding to conformal weight $\frac{1}{2}(m-1)$ and a twistor operator corresponding to conformal weight $-\frac{1}{2}$.

As has been argued in \cite{Tensor}, the natural generalizations of the spinor spaces to the supersetting are infinite dimensional representations. This has also been justified in this paper by the resemblance between the super Dirac operator and the classical Dirac operator. In fact, highest weight representations of $\mathfrak{osp}(m|2n)$ corresponding to the double cover of the supergroup are always infinite dimensional due to the corresponding statement for $\mathfrak{sp}(2n)$. Therefore a proper generalization of the results in \cite{MR0482879} will have to contain infinite dimensional highest weight representations of $\mathfrak{osp}(m|2n)$. This was also the case for the classification in \cite{MR2458281}. There an appropriate class of infinite dimensional $\mathfrak{sp}(2n)$-representations was defined and studied.

On superspace, every conformally invariant operator on $\mR^{ml2n}$, 
\[D:\cO(\mR^{m|2n})\otimes L^{m|2n}_\lambda\to \cO(\mR^{m|2n})\otimes V,\]
for some $\osp$-module $V$ in particular has to be $\mathfrak{osp}(m|2n)$-invariant. By restricting to the first order polynomials in $\cO(\mR^{m|2n})\otimes L^{m|2n}_\lambda$ and only considering the scalar part inside $\cO(\mR^{m|2n})\otimes V$ we obtain that $D$ reduces to an $\osp$-module morphism
\begin{eqnarray*}
\Phi&:&\mC^{m|2n}\otimes L^{m|2n}_\lambda\to V.
\end{eqnarray*}
The existence of such a morphism for a representation $V$ is equivalent to the property 
\[
V\cong\left(\mC^{m|2n}\otimes L_\lambda^{m|2n}\right)/R,
\]
for some $\osp$-module $R\subset \mC^{m|2n}\otimes L_\lambda^{m|2n}$, which corresponds to Ker$(\Phi)$.

As an example of the Fegan classification on superspace we considered the Dirac operator. The corresponding morphism $\Phi$ was constructed in Lemma \ref{Eperp}. We can consider all the differential operators acting between $\cO(\mR^{m|2n})\otimes\mS_{m|2n}^{(\pm)}$ and functions with values in some other representation. By the considerations above, the only candidates are the Dirac operator and the operator corresponding to projection onto the Cartan product. The case $M=m-2n=0$ is different since Theorem \ref{decomp1} implies that the relevant tensor product is not completely reducible. The Dirac operator can still be defined for all values of $M$ as the projection of the gradient on the subrepresentation of the tensor product, isomorphic to a spinor space. This is done in Definition \ref{defDirac}. From Theorem \ref{confSymm} it follows that the super Dirac operator is conformally invariant and the conformal weight is given by $\frac{1}{2}(M-1)$, which is the dimensional c
 ontinuat
 ion of the classical value $\frac{1}{2}(m-1)$. 

The question which needs to be explored as part of the Fegan classification in superspace is whether there is also a conformally invariant differential operator corresponding to the invariant projection of the gradient onto the Cartan product in Theorem \ref{decomp1}. The most interesting case is $M=0$. Based on the $\osp$-invariance and Corollary \ref{tensornotcr} the only other candidate is $\left(\mC^{2n|2n}\otimes \mS_{2n|2n}^\pm \right)/V^\pm$, with $V^\pm$ defined in Theorem \ref{decomp1}, which however is an indecomposable but reducible representation. The quotient irreducible representation is not a candidate because $\left(\mC^{2n|2n}\otimes \mS_{2n|2n}^+\right)/R$ is never equal to $K^{2n|2n}_{\epsilon_1+\omega_n-\frac{1}{2}\nu_n}$ for any subrepresentation $R\subset \mC^{2n|2n}\otimes \mS_{2n|2n}^+$, as follows from Corollary \ref{tensornotcr}. So either reducible representations need to be considered or the classification of invariant differential operators will y
 ield les
 s operators than expected.

Keeping in mind the observed dimensional continuation property and the two classical conformal weights $(m-1)/2$ and $-1/2$, it can be expected that the two conformal weights coincide exactly for the special value $M=0$. This conjecture is also supported by the fact that in the classical case, the conformal weight can be calculated from the value of the quadratic Casimir operator operator of $\mathfrak{so}(m)$ on $L_\mu^m$. Since the two relevant subrepresentations of $\mC^{2n|2n}\otimes \mS^+_{2n|2n}$ correspond to one indecomposable representation $V^+$, they have identical eigenvalues for the Casimir operators of $\osp$.

In the classical case, the appearing conformal weights for the differential operators starting from a certain function space are always strictly different. Because of the arguments above it seems reasonable that this does no longer hold on superspace and that this is intimately related to the appearance of not completely reducible tensor products. The results in Remark \ref{casimir} in case $d=n$ are crucial for this part of the classification. 

\subsection*{Acknowledgment}
The authors would like to thank Vladimir Sou{\v{c}}ek and Dimitry Leites for interesting discussions.

%%   No changes. Try to follow several easy rules while typesetting your
%% bibliography: 1) Paper name should be itacised, 2) use comma (,) as
%% separator, 3) issue number should be boldface.

\end{document}